\documentclass[twoside,11pt]{article}

\usepackage{blindtext}

\usepackage[utf8]{inputenc} 
\usepackage[T1]{fontenc}    
\usepackage{url}            
\usepackage{booktabs}       
\usepackage{amsfonts}       
\usepackage{nicefrac}       
\usepackage{microtype}      
\usepackage{xcolor}         
\usepackage{aligned-overset}
\usepackage{amssymb}
\usepackage{mathtools}
\usepackage{microtype, float, wrapfig}  
\usepackage{amsthm,amsmath,amssymb,bbm}
\usepackage{multirow}
\usepackage{algorithm}
\usepackage{algpseudocode}
\usepackage{natbib}
\usepackage{multibib}

\newcites{supp}{References of Appendix}
\usepackage[preprint]{jmlr2e}
 \hypersetup{ hidelinks }

\usepackage{lastpage}

\ShortHeadings{Knowledge Cascade: Reverse Knowledge Distillation}{Fang, Lu, Chen, Zhong and Ma}

\firstpageno{1}

\begin{document}




\renewcommand{\hat}{\widehat}
\newcommand{\bfm}[1]{\ensuremath{\mathbf{#1}}}
\newcommand{\prob}{\text{P}}
\newcommand{\red}{\textcolor{red}}

\newcommand{\blue}{\textcolor{blue}}
\newcommand{\green}{\textcolor{green}}
\newcommand{\sla}{{\scriptscriptstyle\langle}}
\newcommand{\sra}{{\scriptscriptstyle\rangle}}

\def \bsbeta{\bfsym \beta}
\def \bstheta{\bfsym \theta}
\def\ba{\bfm a}   \def\bA{\bfm A}  \def\AA{\mathbb{A}}
\def\bb{\bfm b}   \def\bB{\bfm B}  \def\BB{\mathbb{B}}
\def\bc{\bfm c}   \def\bC{\bfm C}  \def\CC{\mathbb{C}}
\def\bd{\bfm d}   \def\bD{\bfm D}  \def\DD{\mathbb{D}}
\def\be{\bfm e}   \def\bE{\bfm E}  \def\EE{\mathbb{E}}
\def\bff{\bfm f}  \def\bF{\bfm F}  \def\FF{\mathbb{F}}
\def\bg{\bfm g}   \def\bG{\bfm G}  \def\GG{\mathbb{G}}
\def\bh{\bfm h}   \def\bH{\bfm H}  \def\HH{\mathbb{H}}
\def\bi{\bfm i}   \def\bI{\bfm I}  \def\II{\mathbb{I}}
\def\bj{\bfm j}   \def\bJ{\bfm J}  \def\JJ{\mathbb{J}}
\def\bk{\bfm k}   \def\bK{\bfm K}  \def\KK{\mathbb{K}}
\def\bl{\bfm l}   \def\bL{\bfm L}  \def\LL{\mathbb{L}}
\def\bm{\bfm m}   \def\bM{\bfm M}  \def\MM{\mathbb{M}}
\def\bn{\bfm n}   \def\bN{\bfm N}  \def\NN{\mathbb{N}}
\def\bo{\bfm o}   \def\bO{\bfm O}  \def\OO{\mathbb{O}}
\def\bp{\bfm p}   \def\bP{\bfm P}  \def\PP{\mathbb{P}}
\def\bq{\bfm q}   \def\bQ{\bfm Q}  \def\QQ{\mathbb{Q}}
\def\br{\bfm r}   \def\bR{\bfm R}  \def\RR{\mathbb{R}}
\def\bs{\bfm s}   \def\bS{\bfm S}  \def\SS{\mathbb{S}}
\def\bt{\bfm t}   \def\bT{\bfm T}  \def\TT{\mathbb{T}}
\def\bu{\bfm u}   \def\bU{\bfm U}  \def\UU{\mathbb{U}}
\def\bv{\bfm v}   \def\bV{\bfm V}  \def\VV{\mathbb{V}}
\def\bw{\bfm w}   \def\bW{\bfm W}  \def\WW{\mathbb{W}}
\def\bx{\bfm x}   \def\bX{\bfm X}  \def\XX{\mathbb{X}} 
\def\bY{\bfm Y}   \def\YY{\mathbb{Y}}
\def\bz{\bfm z}   \def\bZ{\bfm Z}  \def\ZZ{\mathbb{Z}}

\def\calA{{\cal  A}} \def\cA{{\cal  A}}
\def\calB{{\cal  B}} \def\cB{{\cal  B}}
\def\calC{{\cal  C}} \def\cC{{\cal  C}}
\def\calD{{\cal  D}} \def\cD{{\cal  D}}
\def\calE{{\cal  E}} \def\cE{{\cal  E}}
\def\calF{{\cal  F}} \def\cF{{\cal  F}}
\def\calG{{\cal  G}} \def\cG{{\cal  G}}
\def\calH{{\cal  H}} \def\cH{{\cal  H}}
\def\calI{{\cal  I}} \def\cI{{\cal  I}}
\def\calJ{{\cal  J}} \def\cJ{{\cal  J}}
\def\calK{{\cal  K}} \def\cK{{\cal  K}}
\def\calL{{\cal  L}} \def\cL{{\cal  L}}
\def\calM{{\cal  M}} \def\cM{{\cal  M}}
\def\calN{{\cal  N}} \def\cN{{\cal  N}}
\def\calO{{\cal  O}} \def\cO{{\cal  O}}
\def\calP{{\cal  P}} \def\cP{{\cal  P}}
\def\calQ{{\cal  Q}} \def\cQ{{\cal  Q}}
\def\calR{{\cal  R}} \def\cR{{\cal  R}}
\def\calS{{\cal  S}} \def\cS{{\cal  S}}
\def\calT{{\cal  T}} \def\cT{{\cal  T}}
\def\calU{{\cal  U}} \def\cU{{\cal  U}}
\def\calV{{\cal  V}} \def\cV{{\cal  V}}
\def\calW{{\cal  W}} \def\cW{{\cal  W}}
\def\calX{{\cal  X}} \def\cX{{\cal  X}}
\def\calY{{\cal  Y}} \def\cY{{\cal  Y}}
\def\calZ{{\cal  Z}} \def\cZ{{\cal  Z}}
\def\bzero{\bfm 0}

\def\E{\text{E}}
\def\R{\mathbb{R}}
\def\cov{\mathrm{cov}}
\def\pcov{\mathrm{Pcov}}
\def\pc{\mathrm{PC}}
\def\a\cos{\mathrm{arc\cos}}

\newcommand{\bfsym}[1]{\ensuremath{\boldsymbol{#1}}}
\def\bsw{\bfsym w}
\def\bsu{\bfsym u}
\def\bsd{\bfsym d}
\def\bsc{\bfsym c}
\def\bss{\bfsym s}
\def\bst{\bfsym t}
\def\bsx{\bfsym x}
\def\bsX{\bfsym X}
\def\bsy{\bfsym y}
\def\bsY{\bfsym Y}
\def\bsz{\bfsym z}
\def\bsI{\bfsym I}

\def\balpha{\bfsym \alpha}
\def\bbeta{\bfsym \beta}
\def\bgamma{\bfsym \gamma}             \def\bGamma{\bfsym \Gamma}
\def\bdelta{\bfsym {\delta}}           \def\bDelta {\bfsym {\Delta}}
\def\bfeta{\bfsym {\eta}}              \def\bfEta {\bfsym {\Eta}}
\def\bmu{\bfsym {\mu}}                 \def\bMu {\bfsym {\Mu}}
\def\bnu{\bfsym {\nu}}
\def\bstheta{\bfsym {\theta}}           \def\bsTheta {\bfsym {\Theta}}
\def\beps{\bfsym \varepsilon}          \def\bepsilon{\bfsym \varepsilon}
\def\bsigma{\bfsym \sigma}             \def\bSigma{\bfsym \Sigma}
\def\blambda {\bfsym {\lambda}}        \def\bLambda {\bfsym {\Lambda}}
\def\bomega {\bfsym {\omega}}          \def\bOmega {\bfsym {\Omega}}
\def\brho   {\bfsym {\rho}}
\def\btau{\bfsym {\tau}}
\def\bpsi{\bfsym {\psi}}               \def\bPsi{\bfsym {\Psi}}
\def\bxi{\bfsym {\xi}}
\def\bzeta{\bfsym {\zeta}}
\def \bkappa {\bfsym{\calK}}
\def \bphi {\bfsym{\varphi}}

 \def\halpha{\hat{\alpha}}              \def\hbalpha{\hat{\bfsym \alpha}}
 \def\hbeta{\hat{\beta}}                \def\hbbeta{\hat{\bfsym \beta}}
 \def\hgamma{\hat{\gamma}}              \def\hgamma{\hat{\bfsym \gamma}}
 \def\hGamma{\hat{ \Gamma}}             \def\hbGamma{\hat{\bfsym \Gamma}}
 \def\hdelta{\hat{\delta}}              \def\hbdelta{\hat{\bfsym {\delta}}}
 \def\hDelta {\hat{\Delta}}             \def\hbDelta{\hat{\bfsym {\Delta}}}
 \def\heta{\hat {\eta}}                 \def\hbfeta {\hat{\bfsym {\eta}}}
 \def\hmu{\hat{\mu}}                    \def\hbmu {\hat{\bfsym {\mu}}}
 \def\hnu{\hat{\nu}}                    \def\hbnu {\hat{\bfsym {\nu}}}
 \def\htheta{\hat {\theta}}             \def\hbtheta {\hat{\bfsym {\theta}}}
 \def\hTheta{\hat {\Theta}}             \def\hbTheta {\hat{\bfsym {\Theta}}}
 \def\hbeps{\hat{\bfsym \varepsilon}}   \def\hbepsilon{\hat{\bfsym \varepsilon}}
 \def\hsigma{\hat{\sigma}}              \def\hbsigma{\hat{\bfsym \sigma}}
 \def\hSigma{\hat{\Sigma}}              \def\hbSigma{\hat{\bfsym \Sigma}}
 \def\hlambda{\hat{\lambda}}            \def\hblambda{\hat{\bfsym \lambda}}
 \def\hLambda{\hat{\Lambda}}            \def\hbLambda{\hat{\bfsym \Lambda}}
 \def\homega {\hat {\omega}}            \def\hbomega {\hat{\bfsym {\omega}}}
 \def\hOmega {\hat {\omega}}            \def\hbOmega {\hat{\bfsym {\Omega}}}
 \def\hrho   {\hat {\rho}}              \def\hbrho {\hat{\bfsym {\rho}}}
 \def\htau   {\hat {\tau}}              \def\hbtau {\hat{\bfsym {\tau}}}
 \def\hxi{\hat{\xi}}                    \def\hbxi{\hat{\bfsym {\xi}}}
 \def\hzeta{\hat{\zeta}}                \def\hbzeta{\hat{\bfsym {\bzeta}}}

\def \T{\mathrm{\scriptstyle T}} 
\def\FDP{\mathrm{FDP}}
\def\wt{\widetilde}
\newcommand\independent{\protect\mathpalette{\protect\independenT}{\perp}}
\def\independenT#1#2{\mathrel{\rlap{$#1#2$}\mkern2mu{#1#2}}}

\title{Knowledge Cascade: Reverse Knowledge Distillation on Nonparametric Multivariate Functional Estimation}

\let \footnote \thanks
\author{\name Luyang Fang\footnote{Equal contribution.}  \email luyang.fang@uga.edu \vspace{6pt}\\
       \name Haoran Lu\footnotemark[\value{footnote}] \email haoran.lu@uga.edu \vspace{6pt}\\
        \name Yongkai Chen \email yongkai.chen@uga.edu \vspace{6pt}\\     
        \name Wenxuan Zhong\footnote{Corresponding authors.} \email wenxuan@uga.edu \vspace{6pt}\\
       \name Ping Ma\footnotemark[\value{footnote}] \email pingma@uga.edu \vspace{6pt}\\
       \addr Department of Statistics\\
       University of Georgia, Athens, GA, 30602, USA}

\editor{Ji Zhu}

\maketitle

\begin{abstract}%
As machine learning models and datasets continue to grow, developing complex models has become increasingly computationally demanding. Knowledge distillation reduces deployment cost by compressing a large, well-trained teacher model into a compact student model, but it does not address settings where constructing the teacher itself is the bottleneck. Motivated by this challenge, we introduce Knowledge Cascade (KCas), a reverse knowledge distillation framework that uses information from a small, inexpensive student model to guide the development of a more complex teacher model. Although this direction is counterintuitive because the teacher typically has greater representational capacity, we show that student-to-teacher transfer can be principled when supported by statistical scaling relationships.
We first develop KCas for nonparametric multivariate functional estimation in reproducing kernel Hilbert spaces via smoothing splines, where selecting multiple smoothing parameters is a major computational bottleneck. KCas transfers student-selected smoothing parameters to the full-sample regime through asymptotic scaling laws, substantially reducing computational cost for high-dimensional and large-scale datasets while retaining theoretical guarantees. Beyond smoothing splines, we illustrate the same principle through kernel density estimation and deep learning hyperparameter transfer. Simulations and real-data experiments show that KCas achieves substantial computational savings while maintaining strong statistical performance, and can sometimes outperform the corresponding full-sample procedure.

\end{abstract}


\section{Introduction}


Machine learning models have recently seen a rapid increase in complexity, particularly with the success of large-scale deep neural networks across language, vision, multimodal learning, and generative modeling \citep{hinton2012deep,wolf2020transformers,hoffmann2022training,dehghani2023scaling,openai2023gpt4,grattafiori2024llama}. Despite their effectiveness, developing such models often requires substantial computational resources and human effort, especially in the presence of large data volumes, high dimensionality, and complex data structures \citep{strubell2019energy,hoffmann2022training,fang2025spot}. As models and datasets continue to grow, it becomes increasingly important to develop principled strategies that use inexpensive preliminary models to guide the development of more complex models. 


Traditionally, knowledge distillation (KD) has served as a widely used framework for mitigating computational costs \citep{hinton2015distilling}. 
In the standard KD paradigm, a large and well-trained teacher model transfers information to a smaller student model, allowing the student to retain much of the teacher's predictive performance with lower computational cost \citep{zhang2019your,gou2021knowledge,fang2024bayesian,ma2026generalizable,fang2026statistical,fang2026knowledge}. While this teacher-to-student paradigm has been highly effective for model compression and efficient deployment, it leaves a critical bottleneck unresolved: standard KD relies on the prerequisite of a fully trained, highly capable teacher model, offering no relief for the exorbitant initial cost of developing the teacher itself.

To address this fundamental limitation in the training phase, we ask whether the mechanics of knowledge transfer can be inverted.
Specifically, can a small and computationally efficient student model provide useful information for \emph{training} a larger teacher model? We propose \emph{Knowledge Cascade} (KCas), a reverse knowledge distillation framework in which information learned from a student model is transferred upward to guide the teacher model. In KCas, the student is not intended to replace the teacher; rather, it serves as a low-cost preliminary model that extracts useful information for the teacher at a substantially reduced computational cost. While transferring knowledge upward to a model with strictly greater representational capacity may initially seem counterintuitive, the key observation behind KCas is that the student does not need to learn everything the teacher will eventually learn. It only needs to extract transferable information that remains structurally useful when moving from a simpler setting to a more complex one. This dynamic is analogous to a statistical pilot study, where a small preliminary experiment cannot replace the full study but provides valuable guidance for conducting it more efficiently. KCas exploits this principle by learning  from a low-cost student model and transferring the insights to the teacher through statistical theory or empirical scaling laws.

We formalize this principle first within the context of nonparametric multivariate functional estimation in reproducing kernel Hilbert spaces (RKHS). In this setting, a subsample-based estimator naturally serves as the student and the full-sample estimator as the teacher, since the effective complexity of the estimator grows with sample size \citep{gu2013smoothing}. The main computational bottleneck is selecting multiple smoothing parameters, whose number and search cost increase rapidly with the number of predictors and interaction terms.
KCas addresses this bottleneck by selecting smoothing parameters on a small subsample and transferring them to the full sample through asymptotic sample-size scaling. The teacher model is then fitted on the full data using these transferred parameters, avoiding costly full-sample tuning while retaining the statistical benefit of using all observations. Our theory and experiments show that KCas substantially reduces computation and can sometimes improve estimation accuracy by avoiding unstable or overly adaptive full-sample tuning.

While smoothing spline ANOVA models provide the primary theoretical anchor for this paper, they represent just one concrete instance of a fundamental principle: whenever an appropriate scaling relationship exists,  information learned from a low-cost student model can be transferred to guide a more expensive teacher model. To demonstrate this versatility, we extend the KCas to kernel density estimation, where bandwidths follow classical asymptotic scaling laws, and to deep learning, where hyperparameters selected on compact student networks can be used to guide larger teacher networks. Ultimately, these extensions position KCas not merely as a specialized technique for RKHS, but as a general, theoretically grounded paradigm for reverse knowledge transfer across diverse machine learning domains. 


\vspace{6pt}
\noindent \textbf{Our contributions:}
\begin{enumerate}
    \item We introduce Knowledge Cascade (KCas), a reverse knowledge distillation framework in which small, computationally efficient student models guide larger and more complex teacher models, providing a new perspective beyond conventional teacher-to-student distillation.
    \item We show that student-to-teacher knowledge transfer can be principled when supported by asymptotic scaling laws, using nonparametric multivariate functional estimation theory to extrapolate information from student models to teacher models.
    \item We develop KCas for smoothing spline ANOVA models, design the associated algorithm, and establish consistency theory. Under the proposed subsampling scheme, KCas reduces the computational complexity of smoothing-parameter selection from $O(n^3)$ to $O(n^{\frac{3}{4}})$ while preserving statistical accuracy.
    \item We illustrate the broader applicability of KCas beyond smoothing splines through kernel density estimation and deep learning hyperparameter transfer.
    \item Through extensive simulations and real-data experiments, we show that KCas achieves substantial computational savings and can even outperform full-sample tuning in some settings.
\end{enumerate}




\section{Related Work}


An idea related to KCas is self-distillation (SD), which is also developed to deploy complex models effectively. SD methods use the same architecture for both the teacher and student models, and facilitate the training procedure by letting the knowledge be transferred/exchanged among a group of models or within a single model \citep{zhang2020self,mobahi2020self,zhang2019your,phuong2019distillation,yang2019snapshot,hou2019learning,lan2018self}. 
However, models in standard SD methods are still relatively large, and the model training procedure is accelerated and improved by distilling knowledge from itself. In this sense, KCas differs from SD by using information from some `actually small' student models that are much easier to train.
In scientific scenarios where the pilot study is needed to determine the design of extensive and detailed experiments, KCas can use the pilot data to construct student models to avoid wasting valuable data in the pilot study.
Thus, KCas is highly desirable in these settings.
Note that in \citet{yuan2020revisiting}, the authors also discussed the possibility of reversing the knowledge distillation procedure, but their methodology is still under the SD framework. \citet{yuan2020revisiting} reverses the KD procedure as a motivating example for proposing the Teacher-free Knowledge Distillation (Tf-KD) framework.
They prove the equivalence between KD and label smoothing regularization in a certain sense, and using this fact, Tf-KD lets a student model learn from itself or manually designed regularization distribution.
Therefore, the student model in Tf-KD serves the purpose of regularization, while the student model in KCas serves the role of extracting information, and KCas amplifies this information to help the teacher model.
Thus, our proposed KCas and the associated theories are significantly different from \citet{yuan2020revisiting} and various self-distillation methods. 

Further, \citet{xie2020self} proposes a similar idea with KCas of distilling the knowledge from the student model to help train the teacher model. They first train a noisy student model on the available labeled data, which is then used to generate pseudo-labels for the unlabeled data. Subsequently, the teacher model is then trained on the combined set of labeled and pseudo-labeled data. On the contrary, our research focuses on addressing situations where the model training process requires an exceptionally high computational burden or when completing the necessary computations becomes impractical. Therefore, \citet{xie2020self} addresses scenarios where the training model has sufficient computational resources but insufficient data, whereas our work focuses on situations where there is abundant data but limited computational resources available for training the model.

From a methodological perspective, KCas can be viewed as a strategy for efficient hyperparameter selection. Recent advancements in hyperparameter selection techniques continue to refine the efficacy of model selection. General approaches of hyperparameter selection include the classic cross-validation \citep{stone1974cross, stone1978cross}, criterion-based methods \citep{akaike1998information, schwarz1978estimating}, and a comprehensive review of recent developments can be found in \citet{bischl2023hyperparameter}. In the context of smoothing spline models, a series of methods based on generalized cross-validation (GCV) has been developed \citep{gu2014smoothing,gu2001cross,hall1990using,wahba1985comparison,gu1991minimizing}. In this paper, we propose to conduct hyperparameter selection via the idea of knowledge cascade.

\section{Preliminaries}
In this section, we provide the necessary background of nonparametric functional estimation in reproducing kernel Hilbert space before introducing our proposed method.
To estimate a function of interest $\eta$ on a generic domain $\mathcal{X}$, we consider the nonparametric penalized loss functional,
\begin{equation}\label{PLik}
   PL = L(\eta)+ {\lambda} J(\eta),
\end{equation}
where
$L(\eta)$ is the goodness-of-fit (loss) functional and $J(\eta)$ is the smoothness (penalty) functional. The smoothing parameter $\lambda$ controls the trade-off between the smoothness of $\eta(x)$ and its fidelity to the data.


\noindent   {\textbf{Functional ANOVA Decomposition.}}
The estimation of function $\eta$ on the product domain $\mathcal{X}=\prod_{j=1}^{d} \mathcal{X}_{j}$ has long been a crucial problem in statistical learning, with numerous proposed methods over the years \citep{jeon2006effective, chen2016comprehensive, lin2006component, perez2009bayesian, bosq2012nonparametric}. 
However, most of them have been challenged by the curse of dimensionality, as the estimation of multivariate functions is intrinsically difficult. To overcome this challenge, one effective approach is to decompose multivariate functions using techniques similar to the classical analysis of variance (ANOVA) decomposition and its associated notions of the main effect and interaction \citep{gu2013nonparametric, gu2003penalized, kim2004smoothing, huang1998projection, jeon2006effective}. 
In this functional ANOVA model, higher-order interactions are often excluded in practical estimation to control model complexity; excluding all interactions yields the popular additive models.
On the product domain $\mathcal{X}=\prod_{j=1}^{d} \mathcal{X}_{j}$, the function $\eta$ can be decomposed as a sum of a constant term, one-dimensional functions (main effects), two-dimensional functions (two-way interactions), and so on, as in the following decomposition,
\begin{equation}\label{eq:decomposition}
    \eta(x)=\eta\left(x_{\langle 1\rangle}, \ldots, x_{\langle d\rangle}\right)=\eta_{\emptyset}+\sum_{j} \eta_{j}\left(x_{\langle j\rangle}\right)+\sum_{j<k} \eta_{j, k}\left(x_{\langle j\rangle}, x_{\langle k\rangle}\right)+\ldots,
\end{equation}
with the constant in $\eta_{\emptyset}$, the main effects in $\eta_{j}$, the two-way interactions in $\eta_{j, k}$, etc. Higher-order interactions are eliminated to ease the curse of dimensionality.

\noindent   {\textbf{Reproducing Kernel Hilbert Space.}}
By adding the smoothness penalty $J(\eta)$ to $L(\eta)$ in loss (\ref{PLik}), we consider the space $\mathcal{H} \subseteq\{\eta: J(\eta)<$ $\infty\}$ in which $J(\eta)$ is a square semi-norm with a finite-dimensional null space $\mathcal{N}_{J}=\{\eta: J(\eta)=0\}$. To assist analysis and computation, a metric and geometry should be defined in this space, and the loss (\ref{PLik}) needs to be continuous in $\eta$ under this metric. Since the reproducing kernel Hilbert spaces (RKHSs) are well suited for this purpose, we consider the space $\mathcal{H}$ as an RKHS with the continuous evaluation $[x] f=f(x)$, reproducing kernel $R(\cdot, \cdot)$, a non-negative definite function satisfying $\langle R(x, \cdot), f(\cdot)\rangle=f(x), \forall f \in \mathcal{H}$, where $\langle\cdot, \cdot\rangle$ is the inner product in $\mathcal{H}.$ 
The existence of the minimizer in RKHS is guaranteed by \citet{wahba1990spline}. Details are discussed in Appendix \ref{Exist}.

\noindent\textbf{Density Estimation.}
Consider the situation that we have independently identically distributed (i.i.d.) data points $x_i,\ i=1,\cdots n$, from an underlying data distribution $p(x)$ on a bounded domain $\mathcal{X}=\prod_{j=1}^{d} \mathcal{X}_{j}$. We aim to estimate $p(x)$ based on observations $x_i$. For the nonparametric setting, a naive maximum likelihood density estimation is meaningless without any nonintrinsic constraint, since it will fit a sum of delta function spikes at the sample points $x_i$, which is obviously not an appealing estimate when the domain $X$ is continuous. Thus, a penalized likelihood estimate (PLE) is a good candidate. Two intrinsic constraints come from the definition of a probability density that $p(\cdot)\ge 0$ and $\int_{\mathcal{X}} p dx=1$.  Since these two constraints are difficult to handle directly in computation, a common approach \citep{gu1993smoothing, silverman1982est} is to estimate the log-density $\eta(\cdot)$, which is free of the constraints through the transformation 
\begin{equation*}\label{eq:den_trans}
    p(x)= \frac{e^{\eta(x)}}{\int_{\mathcal{X}} e^{\eta(x)} d x}. 
\end{equation*}
\citet{silverman1982est} proposed and studied the theoretical properties of the penalized likelihood over a Hilbert space $\mathcal{H}$:
\begin{equation}\label{PLE}
    -\frac{1}{n} \sum_{i=1}^{n} \eta\left(x_{i}\right) + \log \int_{\mathcal{X}} e^{\eta (x)} d x+ \frac{\lambda}{2} J(\eta). 
\end{equation}
\noindent   {\textbf{Nonparametric Regression.}}
Consider the exponential family with the densities of the form 
\begin{equation*}\label{eq:exp}
    f(y \mid x)=\exp \{\frac{y \eta(x)-h(\eta(x))}{a(\phi)}+c(y, \phi)\},
\end{equation*}
where $a(\cdot)>0$, $h$, and $c$ are known functions, $\eta(x)$ is the regression function via the link $\eta$, and $\phi$ is the parameter that is independent of $x$. Observing $Y_i \mid x_i \sim f(y\mid x_i)$, $i=1,\cdots ,n$, we estimate $\eta(x)$ via the penalized likelihood functional
\begin{equation}\label{eq:PL_reg}
    -\frac{1}{n} \sum_{i=1}^{n}\left\{Y_{i} \eta\left(x_{i}\right)-h\left(\eta\left(x_{i}\right)\right)\right\}+\frac{\lambda}{2} J(\eta),
\end{equation}
where the term $c(y,\phi)$ is dropped as it is independent of $\eta(x)$, and $a(\phi)$ is absorbed into $\lambda$.

\section{Methodology}
Nonparametric penalized estimation of the function of interest $\eta$ is a general question in lots of fields \citep{sun2016statistical,helwig2016smoothing}. However, the computational burden of training complex models limits the applicability of many existing methods to large datasets. To tackle this challenge, we propose a reverse version of knowledge distillation, named knowledge cascade, by cascading the knowledge learned from a student model (trained on a small sample) to the teacher model (trained on a large sample). Specifically, we illustrate our method in the context of nonparametric functional estimation, with two important cases: density estimation and regression functional estimation.

\subsection{Minimizer of the Penalized Loss Functional}

We first introduce the computation for the minimizer of the general penalized loss functional (\ref{PLik}). Consider a tensor sum decomposition of the RKHS $\calH = \calN_J \oplus \calH_{J}$. Let $\{\phi_v\}_{v=1}^M$ be a basis of $\mathcal{N}_J = \{\eta: J(\eta) = 0\}$ and $R_J$ be the reproducing kernel in $H_J$. Taking the ANOVA decomposition (\ref{eq:decomposition}) into consideration, the RKHS $\calH_J$ can be further decomposed into $\calH_J= \oplus_{\beta=1}^g \calH_{\beta}$ with the reproducing kernel $R_J=\sum_{\beta=1}^g \theta_\beta R_\beta $, where $R_{\beta}$ is the reproducing kernel in $\calH_{\beta}$. Here the $\theta_{\beta}$'s are an extra set of smoothing parameters adjusting the contribution of the corresponding components. The minimizer of loss (\ref{PLik}) is achieved in the tensor product reproducing kernel Hilbert space $\mathcal{H}$ with the smoothness penalty $J(\eta)=J(\eta, \eta)=\sum_{\beta=1}^g \theta_\beta^{-1}(\eta, \eta)_\beta$, where $(\eta, \eta)_\beta$ are inner products in  $\mathcal{H}_{\beta}$  with reproducing kernels $R_\beta$.
Without loss of generality, we define $J(\eta)$ with tensor-product cubic splines throughout the paper. To ease the notation, we give an example of a tensor product cubic spline on $[0,1]^2$ in Appendix C, and please refer to \citet{gu2013smoothing} for the explicit forms of $\{R_\beta\}_{\beta=1}^g$ and $J(\eta)$.

According to the Kimeldorf-Wahba representer theorem \citep{wahba1990spline, kimeldorf1971some, wang2011smoothing}, the minimizer of loss (\ref{PLik}) has the following form
\begin{equation}\label{eq:eta_minimizer}
\eta(x) = \sum_{v=1}^M d_v\phi_v(x)+\sum_{i=1}^{n} c_i R_J(x_i,x) = \bphi(x)^{T} \bsd+\bxi(x)^{T} \bsc,
\end{equation}
where $\bsd = (d_1,\cdots,d_M)^T$, $\bsc = (c_1,\cdots,c_n)^T$ are unknown coefficients, $\bphi(x)=(\phi_1, \cdots , \phi_M)^T$, $\bxi(x)=(R_J(x_i,\cdot), \cdots , R_J(x_n,\cdot))^T$ are vectors of functions. 
Taking advantage of the representer theorem, the infinite-dimensional optimization problem is transferred into a finite-dimensional one, thereby facilitating the estimation.

For the density estimation problem, plugging the representer of $\eta(x)$ (\ref{eq:eta_minimizer}) into the penalized likelihood of density estimation (\ref{PLE}), the estimation reduces to the minimization of
\begin{equation}\label{eq:loss_den}
    -\frac{1}{n} \mathbf{1}^{T}(Q \mathbf{c}+S \mathbf{d})+\log \int \exp \left(\bphi(x)^{T} \mathbf{d}+\bxi(x)^{T} \mathbf{c}\right) \mathrm{d} x+ \frac{\lambda}{2} \mathbf{c}^{T} Q \mathbf{c},
\end{equation}
where $Q$ is $n\times n$ with the $(i,j)$th entry $R_J(x_i, x_j)$ and $S$ is $n\times M$ with the $(i,v)$th entry $\phi_v(x_i)$.
Similarly, the minimization of the penalized likelihood functional (\ref{eq:PL_reg}) for regression can be achieved via the minimizer of:
\begin{equation}\label{eq:loss_reg}
    \frac{1}{n} \left( \tilde{\bsY} - S \mathbf{d} -Q \mathbf{c} \right)^T W \left( \tilde{\bsY} - S \mathbf{d} -Q \mathbf{c} \right) + \frac{\lambda}{2} \mathbf{c}^{T} Q \mathbf{c},
\end{equation}
where $W$ is the weight matrix, and the explicit form of $\tilde{\bsY}$ and $W$ can be found in Appendix \ref{Minimizer}.
Fixing smoothing parameters $\lambda$ and $\theta$, we can estimate the coefficients $\bsd$ and $\bsc$ in (\ref{eq:loss_den}) or (\ref{eq:loss_reg}) using Cholesky decomposition \citep{golub2013matrix} or Newton-Raphson method \citep{gu1993smoothing, gu2013smoothing}.

To make the estimation practical, a critical aspect is selecting appropriate values for $\lambda$ and $\theta$ that result in satisfactory performance, since the solution of the penalized loss functional (\ref{PLik}) is sensitive to $\lambda$ and $\theta$ \citep{jeon2006effective, gu2013smoothing}. Smoothing parameters control the trade-off between the smoothness of $\eta(x)$ and its fidelity to the data. Selecting a large smoothing parameter will lead to oversmoothing, while a small one will lead to undersmoothing.

One of the most efficient criteria for selecting the smoothing parameters is generalized cross-validation (GCV) \citep{gu1992cross, gu1991minimizing}, which achieves the selection via cross-validation targeting the Kullback--Leibler (KL) loss.
GCV consists of two main steps: (i) minimizing the KL loss with respect to $\lambda$ for fixed $\theta$; (ii) updating $\theta$ according to the updated $\lambda$. 
However, the computational cost of tuning the parameters, particularly $\lambda$, can be prohibitively high in high-dimensional settings. With all $S$ smoothing parameters tunable, the above iterative algorithm takes $O(Sn^3)$ flops per iteration and needs tens of iterations to converge \citep{gu1991minimizing}. Here the number of smoothing parameters $S$ increases as the number of multi-way interaction terms grows. In particular, $\eta(x)=\eta\left(x_{\langle 1\rangle }, \ldots, x_{\langle d\rangle }\right)=\eta_{\emptyset}+\sum_{j} \eta_{j}\left(x_{\langle j\rangle }\right)+\sum_{j<k} \eta_{j, k}\left(x_{\langle j\rangle }, x_{\langle k\rangle }\right)$, the ANOVA decomposition model (\ref{eq:decomposition}) truncated at two-way interactions contains $S=d+\frac{3}{2}d(d-1)$ smoothing parameters. Considering the computational burden, in the case of a particularly large sample size, it is impractical to apply GCV on the full sample, i.e., train the teacher model directly, to accomplish the regression and density estimation tasks using the smoothing spline ANOVA model. To tackle this problem, we propose the knowledge cascade (KCas) method, which enables the determination of the smoothing parameters without incurring a significant computational burden during estimation.

\subsection{Knowledge Cascade}\label{sec:KCas}
In KCas, we aim to let student models learn the smoothing parameters through optimization, with a significantly lower computational burden compared to the teacher, and then transfer the smoothing parameters to teacher models. We first illustrate the definitions of the student and teacher models in our context of nonparametric functional estimation. The teacher model is a complex model trained on the full sample, and the student model is a simple model trained on a subset of the sample with a size of $b$. In (\ref{eq:eta_minimizer}), we can see that the number of kernels, i.e., $R_J(x_i,x)$, equals the sample size $n$. Since the number of kernels highly affects the representation power of the model, the model complexity differs significantly for large and small sample sizes and thus distinguishes the student and teacher models. 
We first illustrate the simplest version of KCas in the general regression model with additive noise,
\begin{equation}\label{eq:ref}
    Y_i = \eta(x_i) + \epsilon(x_i),
\end{equation}
where $\epsilon(\cdot)$ is the white noise process satisfying $\E(\epsilon(x_i))=0$, $\E(\epsilon(x_i)\epsilon(x_j))=\sigma^2$ if $x_i=x_j$, $\E(\epsilon(x_i)\epsilon(x_j))=0$ otherwise. 
Define Hilbert space $\calH^{(m)}$ by 
\begin{equation*}
\begin{aligned}
    \calH^{(m)}=\left\{\eta: \eta^{(v)}\right. &\text{absolutely continuous for } v=0,1, \ldots, m-1, \eta^{(m)} \in \mathcal{L}_{2}[0,1],\\ &\left.\eta^{(v)}(0)-\eta^{(v)}(1)=0 \text{ for } v=0,1, \ldots, m-1 \right\} ,
\end{aligned}
\end{equation*}
where $\eta^{(v)}= \frac{\mathrm{d}^v \eta}{ \mathrm{d} x^v}$, $\mathcal{L}_{2}[0,1]=\left\{f: \int_0^1 f^2 d x<\infty\right\}$ is the Hilbert space formed by all square-integrable functions, and $m$ is a constant indicating the order of smoothness of $\calH^{(m)}$.
For smoothing splines in $\calH^{(m)}$, the optimal smoothing parameter $\lambda$, ignoring $o(1)$ terms, is
\begin{equation*}\label{eq:lambda}
    \lambda = C n^{-\frac{2 m}{2 mp+1}},
\end{equation*}
where $C$ is an unknown constant depending on unknown function $\eta$ \citep{wahba1977practical} and $p\in [1,2]$ indicates different additional smoothness conditions \citep{wahba1977practical,craven1978smooth, wahba1985comparison}. The estimation of $C$ is infeasible since it depends on the unknown true function $\eta$. However, KCas can infer the information of $C$ from a well-trained subsample model (student) and apply it to the full data model (teacher).
Specifically, notice that the asymptotically optimal $\lambda_{\mathrm{GCV}}^{\mathrm{sub}}$ when sample size equals $b$ is 
\begin{equation*}
    \lambda_{\mathrm{GCV}}^{\mathrm{sub}}(b) = Cb^{-\frac{2 m}{2 mp+1}},
\end{equation*}
for the same $C$. 
We estimate the optimal $\lambda_{\mathrm{GCV
    }}^{\mathrm{sub}}$ on the subsample to infer the constant $C$, and then employ the same $C$ for the full data \citep{sun2021asymptotic}. That is, we have the following estimation of $\lambda$ by KCas,
\begin{equation}\label{KCas}
    \lambda_{\mathrm{KCas}}^{\mathrm{full}}(n;b) = \lambda_{\mathrm{GCV
    }}^{\mathrm{sub}}(b) (\frac{n}{b})^{-\frac{2m}{2mp+1}}.
\end{equation}
Since the smoothing parameters are used to determine the proportion of the smoothness penalty on different terms in (\ref{eq:decomposition}) and this proportion should be stable over different sample sizes, we directly use the optimal $\theta_{\mathrm{GCV}}^{\mathrm{sub}}(b)$ in the full sample. 
We then generalize the estimator (\ref{KCas}) from the regression model with additive noise (\ref{eq:ref}) to a wide range of penalized likelihood estimation problems, including the nonparametric regression in the exponential family and density estimation.

We propose to use KCas to transfer the knowledge of smoothing parameters in (\ref{PLE}) or (\ref{eq:PL_reg}) to the teacher model. The KCas algorithm is summarized in the following Algorithm \ref{algorithm 1}. In the first step, we apply uniform sampling to get a subsample $X_b$, and our experiments show that uniform sampling can achieve good performance. Other more dedicated sampling methods can also be applied to improve the performance further \citep{wang2018optimal, meng2020more,daszykowski2002representative}.
The total number of operations required for each iteration is generally $\frac{4}{3}n^3+O(n^2)$, in which the selection of the smoothing parameter takes the major burden. The GCV algorithm for the student model takes $O(Sb^3)$ flops per iteration, and thus KCas algorithm reduces the computational cost from $O(Sn^3)$ to $O(Sb^3)$. 
According to the justifications in \citet{gu2002penalized,kim2004smoothing,ma2017adaptive,zhang2023optimal}, it is sufficient to take subsample size $O(n^{\frac{2}{9}})$ to maintain the performance of smoothing spline estimation, and we thus take a slightly larger subsample size $b = O(n^{\frac{1}{4}})$ for practical use in our algorithm. Consequently, our proposed KCas method, as detailed in Algorithm \ref{algorithm 1}, substantially lowers the burden of the smoothing parameter estimation process to $O(Sn^{\frac{3}{4}})$. This underscores the pivotal role of KCas.


\begin{algorithm}[H]
\caption{\strut KCas for nonparametric function estimation.}\label{algorithm 1}
\strut \textbf{Input:} Data $X$, subsample size $b$. \
\begin{algorithmic}[1]
\State Use uniform sampling to select a subsample $X_b$ of size $b$ from the full sample $X$ of size $n$. Apply GCV on $X_b$ to estimate the smoothing parameters, denoting by $\lambda_{\mathrm{GCV}}^{\mathrm{sub}}(b)$ and $\theta_{\mathrm{GCV}}^{\mathrm{sub}}(b)$.
\State Get the estimation of smoothing parameters for the full sample $X$ using $\lambda_{\mathrm{KCas}}^{\mathrm{full}}(n;b) = \lambda_{\mathrm{GCV}}^{\mathrm{sub}}(b) (\frac{n}{b})^{-\frac{2m}{2mp+1}}$ and 
$\theta_{\mathrm{KCas}}^{\mathrm{full}}(n;b) = \theta_{\mathrm{GCV}}^{\mathrm{sub}}(b)$.
\State Fit smoothing splines via penalized likelihood on $X$ with $\lambda_{\mathrm{KCas}}^{\mathrm{full}}(n;b)$ and $\theta_{\mathrm{KCas}}^{\mathrm{full}}(n;b)$ to get the function estimation $\hat{\eta}$.
\end{algorithmic}
\strut \textbf{Output:} Estimation $\hat{\eta}$.
\end{algorithm}

In practice, hyperparameters $m$ and $p$ need to be selected properly. For the univariate setting, the commonly used smooth level is $m=2$, which means the penalty is $J(\eta, \eta)=\int_{0}^{1}\left(\eta^{(2)}\right)^{2} \mathrm{~d} x$ on $[0,1]$. For the multivariate setting, we suggest applying the commonly used tensor product cubic spline, which has $2-\epsilon < m < 2, \forall \epsilon >0$ 
\citep{wahba1990spline}.
When $\eta^{(2)}$ is square-integrable, we have $p=1$, and when $\eta^{(4)}$ is square-integrable, we have $p=2$. The selection of $p$ can vary across datasets.
In this work, we take $m=2$ and $p=2$ empirically.



We should mention that GCV is typically used for Gaussian-type regression. For regression with responses from exponential families which results in a nonquadratic loss (\ref{eq:PL_reg}), the computation of $\eta_\lambda$ requires a bunch of iterations even with fixed smoothing parameters and thus results in a more expensive computational cost. 
To address this issue, we use the well-accepted generalized approximate cross-validation (GACV) method \citep{xiang1996generalized,gu2001cross} or its variants to reduce the computational burden. 
Similarly, a variant of GCV \citep{gu2013nonparametric} will be used to effectively select the smoothing parameter.
To make the notation concise, variants of GCV, suitable for different problems, are collectively referred to as `GCV'.


\subsection{Theoretical Analysis}
In this section, we present the theoretical properties of the smoothing parameters $\lambda$ selected according to Algorithm \ref{algorithm 1}. For notational simplicity, in the following, we will use $\lambda$ to represent all the smoothing parameters, not just the $\lambda$ in front of $J(\eta)$. All proofs for this section are relegated to Appendix \ref{proof}.

Denote by $\hat{\eta}$ the estimate through the minimization of (\ref{PLE}) and $\eta_0$ the true function to be estimated, the asymptotic convergence rate based on the selected smoothing parameter $\lambda_{\mathrm{KCas}}^{\mathrm{full}}(n;b)$ is established through Theorem \ref{thm:converge_den}.

\begin{theorem}[rate for density estimates]\label{thm:converge_den}
For the density estimation problem as in (\ref{PLE}), denote $\eta_0$ the true log-density to be estimated.
Under the regularity conditions 
\ref{cond:continuous} to \ref{cond:bound} 
in Appendix \ref{condi}, 
assuming $\lambda_{\mathrm{GCV}}^{\mathrm{sub}}(b) \rightarrow 0$
and $b (\lambda_{\mathrm{GCV}}^{\mathrm{sub}}(b))^{\frac{1}{2m}} \rightarrow \infty$ as $b \rightarrow \infty$,
we have the convergence rate for density estimates,
\begin{equation} \label{eq:thm1}
    (V+ \lambda_{\mathrm{KCas}}^{\mathrm{full}}(n;b)  J)\left(\hat{\eta}-\eta_0\right)=O_{p}\left(n^{-1}  \lambda_{\mathrm{KCas}}^{\mathrm{full}}(n;b) ^{-\frac{1}{2m}}+ \lambda_{\mathrm{KCas}}^{\mathrm{full}}(n;b) ^{p}\right),
\end{equation} 
where $V(\cdot)$ is an interpretable metric such that a small $V(\hat{\eta}-\eta_0)$ indicates a good estimate, $p$ and $m$ are constant parameters specified in Appendix \ref{condi}.
\end{theorem}

With Theorem \ref{thm:converge_den}, the consistency is ensured, and the convergence rate is specified for the estimation $\hat{\eta}$ in density estimation based on KCas. Denote by $\hat{\eta}$ the estimate through the minimization of (\ref{eq:PL_reg}), the following Theorem \ref{thm:converge} further demonstrates the theoretical properties of KCas in the context of nonparametric regression.
\begin{theorem}[rate for regression estimates]\label{thm:converge}
For the regression in exponential families as in (\ref{eq:PL_reg}), denote $\eta_0$ the true function to be estimated.
Under the regularity conditions 
\ref{cond:continuous}, \ref{cond:summation}, \ref{cond:eigenvalue}, \ref{cond:bound}, and \ref{cond:equi}
in Appendix \ref{condi}, 
assuming $\lambda_{\mathrm{GCV}}^{\mathrm{sub}}(b) \rightarrow 0$
and $b (\lambda_{\mathrm{GCV}}^{\mathrm{sub}}(b))^{\frac{1}{m}} \rightarrow \infty$ as $b \rightarrow \infty$,
we have the convergence rate for regression estimates,
\begin{equation}  \label{eq:thm2}
    (V+ \lambda_{\mathrm{KCas}}^{\mathrm{full}}(n;b)  J)\left(\hat{\eta}-\eta_0\right)=O_{p}\left(n^{-1}  \lambda_{\mathrm{KCas}}^{\mathrm{full}}(n;b) ^{-\frac{1}{2m}}+ \lambda_{\mathrm{KCas}}^{\mathrm{full}}(n;b) ^{p}\right), 
\end{equation} 
where $V(\cdot)$, $p$ and $m$ are defined in the same way as in Theorem \ref{thm:converge_den}.
\end{theorem}
\noindent \textbf{Remark.}
Note that the rates (\ref{eq:thm1}) and (\ref{eq:thm2}) are free of dimension $d$, which is because we adopt the dimensionless approach as in Chapter 9 in \citet{gu2013smoothing}, which essentially incorporates assumptions on the smoothness to get rid of the dimension in the rates.
Specifically, Condition \ref{cond:eigenvalue} implicitly puts an assumption on the smoothness of $\eta$, which makes the error rates (\ref{eq:thm1}) and (\ref{eq:thm2}) to be free of dimension $d$. The condition becomes more stringent as dimension $d$ becomes larger.
The results can be adapted to the case with dimension using arguments in \citet{lin2000tensor}.

\subsection{Knowledge Cascade in Kernel Density Estimation}

We next illustrate the KCas framework in the classical univariate kernel density estimation (KDE) problem. Let $x_1,\ldots,x_n$ be i.i.d.\ observations from an unknown density $f$ on $\mathbb{R}$. For a symmetric kernel $K$ and a bandwidth $h>0$, the Gaussian kernel density estimator is
\[
\hat{f}_h(x)
  = \frac{1}{n}\sum_{i=1}^n K_h(x-x_i)
  = \frac{1}{nh} \sum_{i=1}^n K\!\left(\frac{x-x_i}{h}\right), 
\]
where we take $K$ to be the standard normal density throughout \citep{silverman2018density}. For twice differentiable $f$, the classical asymptotic mean integrated squared error (AMISE) theory yields an optimal bandwidth of the form
\[
h_{\mathrm{AMISE}}
  = \Bigg\{\frac{R(K)}{m_2(K)^2\,R\!\big(f''\big)}\Bigg\}^{1/5} n^{-1/5}
  = C\,n^{-1/5},
\]
where $R(g)=\int g^2(x)\,dx$, $m_2(K)=\int x^2 K(x)\,dx$, and $C$ is an unknown constant depending on $f$ and $K$ \citep{silverman2018density,jones1996survey}.

In KCas, we treat a KDE fitted on a small subsample as the student model and the full-sample KDE as the teacher model. The key idea is to estimate the constant $C$ from a subsample of size $b$ and then transfer it to the full data via the AMISE scaling. Specifically, the AMISE-optimal bandwidth for the subsample satisfies
\[
h_{\mathrm{AMISE}}^{\mathrm{sub}}(b) = C\,b^{-1/5},
\]
so that $C = h_{\mathrm{AMISE}}^{\mathrm{sub}}(b)\,b^{1/5}$. Plugging this expression into the full-sample formula gives the KCas bandwidth
\begin{equation}
\label{eq:KCas_KDE}
h_{\mathrm{KCas}}^{\mathrm{full}}(n; b)
  = h_{\mathrm{AMISE}}^{\mathrm{sub}}(b)\!\left(\frac{n}{b}\right)^{-1/5}.
\end{equation}
In our implementation, $h_{\mathrm{AMISE}}^{\mathrm{sub}}(b)$ is not computed by oracle access to $f$ but by a data-driven selector applied to the subsample; the KCas bandwidth then extrapolates this choice to the full data. This is a direct KDE analogue of the smoothing-parameter scaling used in smoothing splines.

The same knowledge cascade principle extends naturally beyond KDE to a broad class of kernel-based methods. In many kernel procedures, such as kernel regression \citep{wand1994kernel}, local polynomial smoothing \citep{fan2018local}, and kernel-based covariance estimation \citep{ferraty2006nonparametric}, the optimal tuning parameters follow explicit sample-size dependent scaling laws derived from asymptotic theories. KCas exploits this structure by learning the problem-dependent constants from a small subsample and transferring them to the full sample through the corresponding scaling relations, thereby reducing computational cost while preserving statistical efficiency. This perspective suggests KCas as a general strategy for tuning-parameter transfer in kernel methods with known asymptotic rates.

\subsection{Knowledge Cascade in Deep Learning}

Although KCas is developed in the context of nonparametric multivariate functional estimation, the underlying principle extends naturally to settings in which computational constraints make full hyperparameter tuning difficult. In many deep learning applications, the cost of identifying stable and high-performing hyperparameter configurations, particularly the learning rate schedule, dominates the total training budget. Motivated by this challenge, we investigate a simple extension of KCas that uses a compact student network to guide the hyperparameter choices for a larger teacher network.

In this setting, the student model is trained on a reduced architecture and its optimal hyperparameters are obtained through standard tuning procedures. KCas then scales these hyperparameters to the teacher model using a batch-size dependent rule. This approach acknowledges empirical findings that learning rate schedules tuned on smaller networks often transfer to larger networks when batch sizes are adjusted appropriately. Consequently, KCas reduces hyperparameter search cost by allowing tuning to be performed once on the student model and then extrapolated to the teacher.

Mathematically, let $\eta_{\text { student }}$ denote the optimal learning rate for the student model, and let $B_{\text {student }}$ and $B_{\text {teacher }}$ be the respective batch sizes. The KCas learning rate for the teacher model is defined as
\begin{equation*}
    \eta_{\text { teacher }}=\eta_{\text { student }} \cdot g\left(\frac{B_{\text {teacher }}}{B_{\text {student }}}\right) ,
\end{equation*}
where $g(\cdot)$ is a monotone scaling function. We consider two commonly used choices. The first is the linear rule $g(r)=r$, inspired by the large-batch scaling observations in \citet{goyal2017accurate}. Since teacher architectures in our experiments are substantially more complex and often benefit from more conservative step sizes, we also examine the square-root rule $g(r)=\sqrt{r}$. These two options illustrate how KCas can adapt scaling behavior to model complexity.

Although this extension does not rely on the RKHS-based theory that supports KCas in nonparametric functional estimation, it follows the same guiding philosophy: use a computationally efficient student model to extract stable hyperparameter information and propagate it upward to a more complex teacher model. As demonstrated in Section \ref{sec_real:dl}, this strategy provides a practical reduction in tuning cost while maintaining competitive predictive performance for modern deep learning architectures.

\section{Simulation Study}

We conduct simulation studies to evaluate the proposed KCas framework from two perspectives. First, we examine its performance in the main setting of this paper: smoothing spline ANOVA models for density estimation and nonparametric regression. These experiments assess whether smoothing parameters learned from a subsample-based student model can be effectively transferred to the full-sample teacher model while reducing computational cost. Second, we investigate KCas in KDE to illustrate the same student-to-teacher transfer principle in a different nonparametric setting where bandwidths follow classical asymptotic scaling laws.

\subsection{Simulation 1: Density Estimation with Smoothing Splines} \label{sec:s1}
We first evaluate the proposed KCas method for smoothing splines on synthetic density estimation problems. We consider the following two data-generating scenarios.
\\
\textbf{Scenario 1:} A $d$-dimensional Gaussian mixture model, 
$\frac{1}{d} \sum_{i=1}^d \text{MVN}(e_i, I_d)$,
where $e_i$ is the vector with the $i$th entry being $1$ and others being $0$. We consider $d=3,6$.
\\
\textbf{Scenario 2:} A $d$-dimensional density is constructed by independently combining a $5$-dimensional Gaussian with mean zero and variance $0.5(\mathbf{1}\mathbf{1}^T)+1.5I_d$, and the remaining $d-5$ variables are i.i.d. from Unif$(0,1)$. We consider $d=15,20$.

We evaluate the methods using the relative Kullback--Leibler (KL) divergence with respect to the benchmark method, defined as 
\begin{equation}
  \operatorname{RelKL}\left(\hat{p}, p^*, p\right) =   \frac{D_{\mathrm{KL}}(\hat{p}\Vert p)}{D_{\mathrm{KL}}(p^* \Vert p)}, \nonumber
\end{equation}
where $p$ is the true distribution, $\hat{p}$ is the estimate from the method being evaluated, and $p^*$ is the estimate from the benchmark method.
Here $D_{\mathrm{KL}}(q\Vert p)$ denotes the KL divergence of $q$ relative to $p$.
A smaller relative KL divergence indicates better performance. In each replication, full-sample GCV is used as the benchmark, and the relative KL divergence is computed with respect to the corresponding GCV estimate from that replication. For visualization, we report log-transformed RelKL.

For the smoothing spline experiments, including Simulation 1 here and following Simulation 2, we compare KCas with the following baseline methods:
\begin{itemize}
    \item GCV \citep{gu2013nonparametric} on the full sample: Full-sample GCV uses all available observations for smoothing-parameter selection and is therefore used as the benchmark. All other methods are evaluated using relative performance measures with respect to full-sample GCV.
    \item GCV on subsample (SUB): GCV is performed on a randomly selected subsample to reduce computational cost.
    \item GCV in generalized additive models (GAM) \citep{wood2004stable}: GCV is performed on the full sample using a simplified additive model with only main effects.
    \item SKIP method \citep{gu2014smoothing}: SKIP accelerates computation by selecting an appropriate starting point and bypassing subsequent iterations. Since it often fails to converge in relatively high-dimensional density estimation settings, we include it only for the nonparametric regression problem.
    \item Order-based method (ORD) \citep{hall1990using}: ORD directly sets the smoothing parameter to $\lambda=n^{-\frac{2m}{2mp+1}}$, where $n$ is the full sample size.
    \item Kernel density estimation (KDE): For density estimation, we further include the KDE method of \citet{nagler2016evading} as a comparison for high-dimensional data.
\end{itemize}

For KCas, we use uniform sampling to select a subsample of size $b=50n^{\frac{1}{4}}$. The full sample sizes are $5,000$, $10,000$, and $20,000$. Following the discussion in Section \ref{sec:KCas}, we set $m=2$ and $p=2$ in practice. All results are based on $30$ replications.

\begin{figure}[ht]
\centering
\includegraphics[width=1\linewidth]{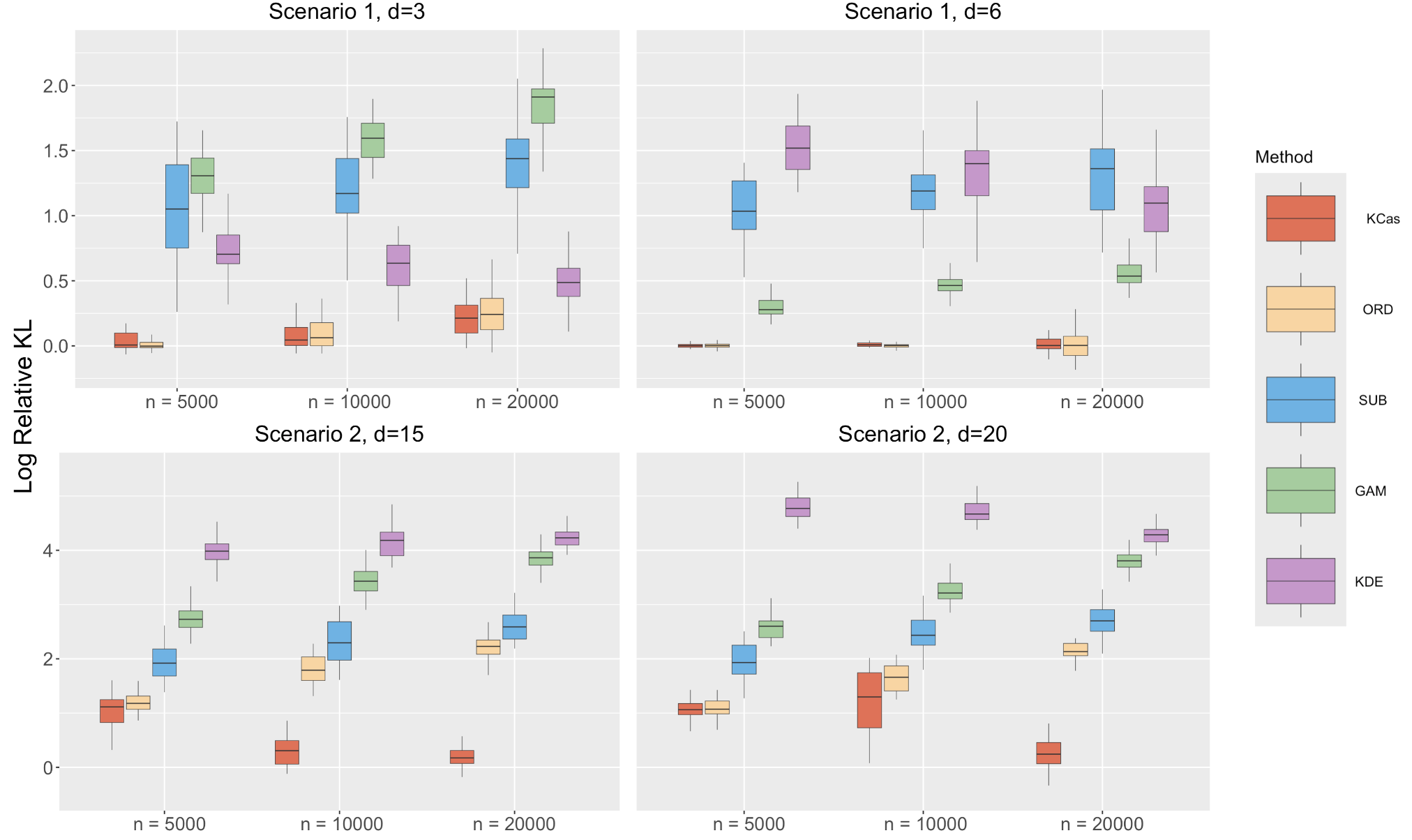}
\caption{Performance comparisons of different methods in density estimation problems using the log-transformed relative KL divergence. The lower relative KL divergence indicates better performance. Two scenarios of data generation processes are provided, each including several different settings of dimension $d$ and sample size $n$.}
\label{simu:ssden}
\end{figure}

Figure~\ref{simu:ssden} illustrates the log-transformed RelKL of KCas and the competing methods.
KCas performs favorably across the considered settings, with log-RelKL values close to or below zero in many cases. Negative values indicate that KCas outperforms the full-sample GCV benchmark in that replication. This suggests that smoothing parameters learned from the student model can effectively guide the teacher model and, in some settings, may also reduce the effect of overly adaptive full-sample tuning.

\subsection{Simulation 2: Nonparametric Regression with Smoothing Splines} \label{sec:s2}
For the nonparametric regression problem, we consider the model
\begin{equation} 
    y_i \sim \text{Ber}\left(\frac{\exp\left(\eta\left(x_i\right)\right)}{1+\exp\left(\eta\left(x_i\right)\right)}\right), \nonumber
\end{equation}
where $\text{Ber}(p)$ denotes the Bernoulli distribution with probability $p$, $x_{i}=\left(x_{i\langle 1\rangle}, \ldots, x_{i\langle d\rangle}\right)^{\mathrm{T}}$ is the $d$-dimensional predictor for the $i$th observation, with each entry independently drawn from Unif$(0,1)$. $\eta$ is the nonparametric function determining the success probability in the Bernoulli trial, and $y_{i} \in \{0,1\}$ is the response variable for the $i$th observation.

We evaluate the methods by the relative MSE, defined by
\begin{equation}
   \text{RelMSE}\left(\hat{\eta}, \eta^*, \eta\right) = \frac{\sum_{i=1}^{n}\left\{\hat{\eta}\left(x_{i}\right)-\eta\left(x_{i}\right)\right\}^{2}}  {\sum_{i=1}^{n}\left\{\eta^{*}\left(x_{i}\right)-\eta\left(x_{i}\right)\right\}^{2}} , \nonumber
\end{equation}
where $\eta$ is the true function, $\hat{\eta}$ is the estimation by the method being evaluated, and $\eta^*$ is the benchmark method. We compare KCas with the same baseline methods as in Simulation 1.

We consider two different scenarios with different dimensions. We report log-transformed RelMSE for plotting clearness, and each RelMSE is computed based on a full GCV baseline for that particular replication. 
\\
\textbf{Scenario 1:}
\begin{align*}
\eta_{m 1}(x) =
  \sum_{i=1}^3 g_1\left(x_{\langle i \rangle}\right) + g_2\left(x_{\langle 1 \rangle}, x_{\langle 2 \rangle}\right) + g_2\left(x_{\langle 1 \rangle}, x_{\langle 3 \rangle}\right) + g_3\left(x_{\langle 1 \rangle}, x_{\langle 2 \rangle}, x_{\langle 3 \rangle}\right).
\end{align*}

\noindent \textbf{Scenario 2:} 
\begin{align*}
\eta_{m 2}(x) &=
  \sum_{i=1}^3 i g_1\left(x_{\langle i \rangle}\right) + \sum_{i=4}^6 i g_{5}\left(x_{\langle i \rangle}\right) + \sum_{i=7}^9 g_{4}\left(x_{\langle i \rangle}\right) + \\ & \sum_{i=1}^3 \sum_{j>i}^4 3i g_2\left(x_{\langle i \rangle}, x_{\langle j \rangle}\right) + 6 g_2\left(x_{\langle 5 \rangle}, x_{\langle 6 \rangle}\right) + 8 g_6\left(x_{\langle 7 \rangle}, x_{\langle 8 \rangle}\right) + 10 g_3\left(x_{\langle 1 \rangle}, x_{\langle 2 \rangle}, x_{\langle 3 \rangle}\right).
\end{align*}

The explicit forms of the functions $g_i$ are provided in Appendix \ref{Appx:sim_gss}. 
Scenario 1 is a well-established setting for nonparametric multivariate functional estimation in RKHS \citep{jeon2006effective, gu1991minimizing, sun2021asymptotic, gu2003penalized}. We consider two cases, with $d=3$ and $d=6$ predictors. When $d=3$, all predictors contribute to $\eta_{m1}$ and hence to the response $y$. When $d=6$, the last three predictors are irrelevant to $\eta_{m1}$, representing a setting with redundant covariates. Scenario 2 is a more complex high-dimensional setting, where we consider $d=15$ and $d=20$.

\begin{figure}[ht]
\centering
\includegraphics[width=1\linewidth]{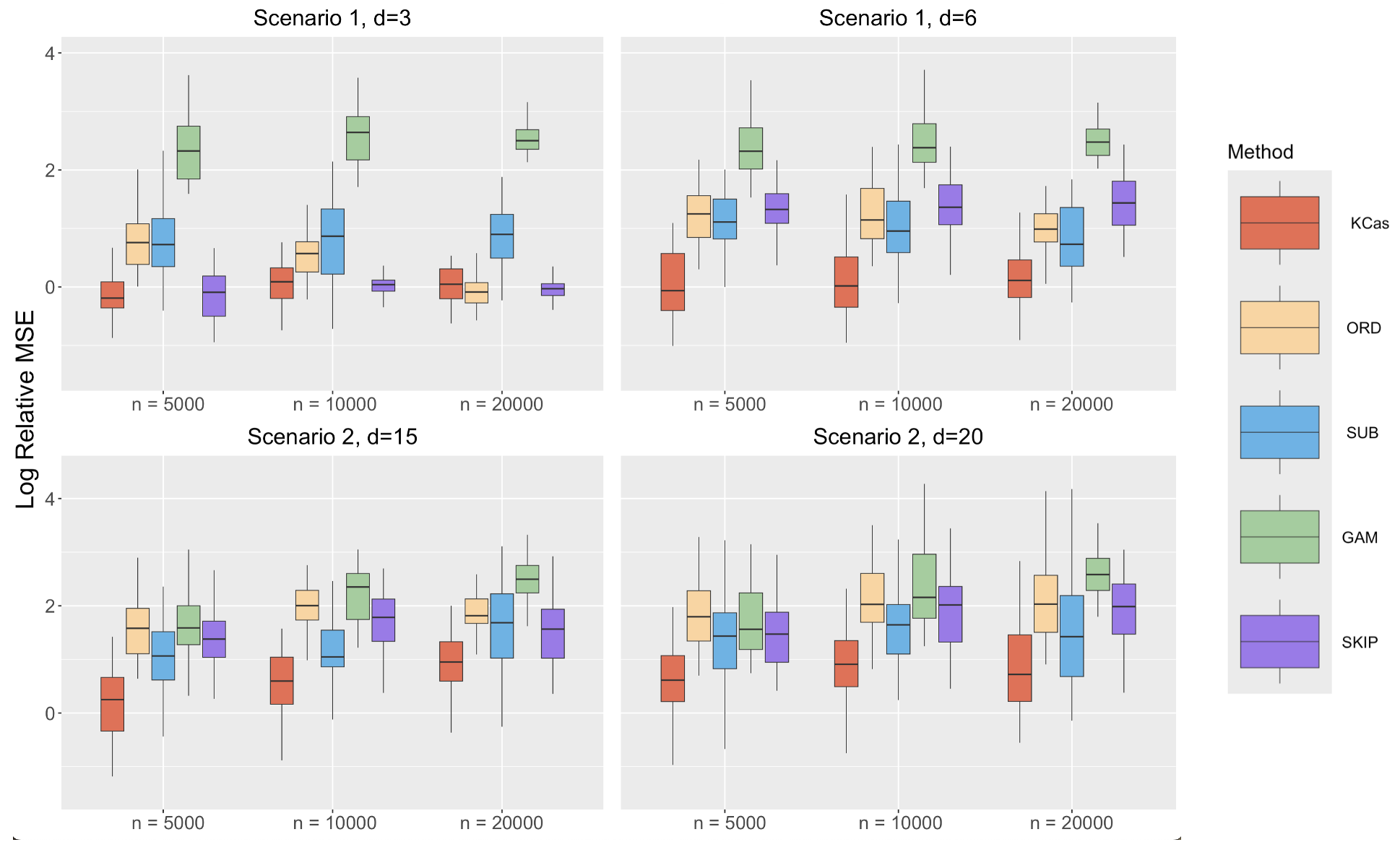}
\caption{Performance comparisons of different methods in nonparametric regression problems using log RelMSE. The lower log RelMSE indicates better performance. Two scenarios of data generation processes are provided, each including several different settings of dimension $d$ and sample size $n$.
}\label{simu:gss}
\end{figure}

\begin{table}[ht]
\centering
\vspace{6pt}
\scalebox{0.9}{
\begin{tabular}{c c c c c c c}
\toprule
& GCV & KCas & GAM &SUB & ORD   & SKIP \\
\midrule
$n=5000$ &42.2 & 6.1 & 1.3 &5.5 &16.2 & 2.1  \\
$n=10000$ & 72.9 &9.4 &2.5 & 7.0 & 23.1 & 3.3  \\
$n=20000$ & 102.5 & 8.1 &3.2 &5.0 &30.8  &3.5 \\         
\bottomrule
\end{tabular}
}
\caption{Comparison of median computational time (min) under the most difficult simulation setting (simulation 2, scenario 2, $d=20$).}
\label{table:gss_time}
\end{table}

Figure \ref{simu:gss} displays the log-transformed RelMSE relative to full-sample GCV. In Scenario 1 with $d=3$, KCas and SKIP show comparable performance and outperform the other methods. Their median log-RelMSE values are close to, and sometimes below, zero, indicating performance comparable to or better than full-sample GCV. When $d=6$, KCas remains close to full-sample GCV, whereas SKIP performs substantially worse, with median log-RelMSE values greater than $1$. This is likely because SKIP uses the starting values introduced by \citet{gu2014smoothing} as the final estimate; when $\eta$ becomes more complex, such starting values may be far from the optimum, leading to inaccurate estimation. In Scenario 2, similar patterns are observed for both $d=15$ and $d=20$, with KCas achieving the best overall performance.

To empirically compare computational efficiency, we report the running time under the most challenging simulation setting, Scenario 2 with $d=20$, in Table~\ref{table:gss_time}. KCas substantially reduces computation time compared with full-sample GCV while maintaining comparable estimation accuracy.
It is important to note that while methods such as GAM, SUB, and SKIP require less computation time than KCas, they do not achieve the same level of accuracy.

\subsection{Simulation 3: Kernel Density Estimation}
\label{sec:sim_kde}

We next conduct an additional simulation to evaluate KCas in classical kernel density estimation. This experiment is separate from the smoothing spline simulations above and uses bandwidth-selection methods as baselines. The goal is to examine whether information learned from a small subsample can be transferred to the full-sample KDE estimator through the classical $n^{-1/5}$ bandwidth scaling rule.

We consider six univariate benchmark densities widely used in KDE studies \citep[e.g.,][]{silverman2018density,jones1996survey}: 
(1) standard normal density $\mathcal{N}(0,1)$; 
(2) symmetric bimodal Gaussian mixture $0.6\,\mathcal{N}(-1,1^2)+0.4\,\mathcal{N}(2,0.5^2)$; 
(3) lognormal density $\text{Lognormal}(0,1)$; 
(4) trimodal Gaussian mixture $0.4\,\mathcal{N}(-2,0.5^2)+0.3\,\mathcal{N}(0,0.6^2)+0.3\,\mathcal{N}(3,0.8^2)$; 
(5) gamma density $\text{Gamma}(2,1.5)$; and 
(6) skewed bimodal Gaussian mixture $0.7\,\mathcal{N}(-1,1^2)+0.3\,\mathcal{N}(3,0.5^2)$. 
For each density, we generate samples of size $n \in \{200, 500, 1000, 2000, 5000\}$. For a given density and $n$, we draw $50$ independent replications and, for each replication, compute several bandwidth selectors and their associated KDEs.

We compare the following methods:
\begin{itemize}
\item[(i)] \textbf{Oracle AMISE}: the theoretical AMISE-optimal bandwidth $h_{\mathrm{AMISE}}$ obtained by numerically evaluating $R\!\big(f''\big)$ from the known $f$. This serves as a lower bound on the achievable mean integrated squared error (MISE).
\item[(ii)] \textbf{Improved Sheather-Jones (ISJ)}: the diffusion-based plugin selector of \citet{botev2010kernel}, which refines the original Sheather--Jones solve-the-equation method \citep{sheather1991reliable} and is known to perform well across a wide range of densities.
\item[(iii)] \textbf{Least-squares cross-validation (LSCV)}: the classical least-squares cross-validation bandwidth \citep{rudemo1982empirical,bowman1984alternative}, which minimizes an unbiased estimate of the integrated squared error.
\item[(iv)] \textbf{Silverman's rule of thumb}: the normal-reference bandwidth $h_{\mathrm{S}} = 0.9\min\{\hat{\sigma},\mathrm{IQR}/1.34\}$ $n^{-1/5}$ \citep{silverman2018density}, where $\hat{\sigma}$ is the sample standard deviation and $\mathrm{IQR}$ is the sample interquartile range, defined as the difference between the third and first sample quartiles, $\mathrm{IQR} = Q_{3} - Q_{1}$.
\item[(v)] \textbf{KCas-Silverman}: a KCas estimator based on the Silverman rule computed on a uniform subsample of size $b=200$, i.e.\ $h_{\mathrm{KCas},S}(n;b)$ obtained by plugging $h_{\mathrm{AMISE}}^{\mathrm{sub}}(b) = h_{\mathrm{S}}^{\mathrm{sub}}(b)$ into \eqref{eq:KCas_KDE}.
\item[(vi)] \textbf{KCas-ISJ}: a KCas estimator based on the ISJ selector computed on a subsample of size $b=200$, i.e.\ $h_{\mathrm{KCas},\mathrm{ISJ}}(n;b)$ from \eqref{eq:KCas_KDE} with $h_{\mathrm{AMISE}}^{\mathrm{sub}}(b)$ taken as the subsample ISJ bandwidth.
\item[(vii)] \textbf{KCas-CV}: a KCas estimator based on least-squares cross-validation (LSCV) computed on a subsample of size $b=200$, i.e.\ $h_{\mathrm{KCas},\mathrm{CV}}(n;b)$ obtained by plugging the subsample LSCV bandwidth $h_{\mathrm{CV}}^{\mathrm{sub}}(b)$ into \eqref{eq:KCas_KDE}.
\end{itemize}
In each method, we fit a Gaussian KDE on an equally spaced grid of $4096$ points covering the effective support of each density. The integrated squared error (ISE) for a replication is approximated via the trapezoidal rule,
\[
\mathrm{ISE}(\hat{f}_h) \approx \sum_{k} \big\{\hat{f}_h(x_k) - f(x_k)\big\}^2 \,\Delta x,
\]
and we summarize performance using the median ISE and interquartile range across the $50$ replications. To facilitate comparison, we also report the ratio of the median ISE to that of the oracle AMISE bandwidth for each density and $n$, as well as log--log slopes of median bandwidth versus sample size to check the $n^{-1/5}$ scaling.

\begin{figure}[h]
    \centering
    \includegraphics[width=1\linewidth]{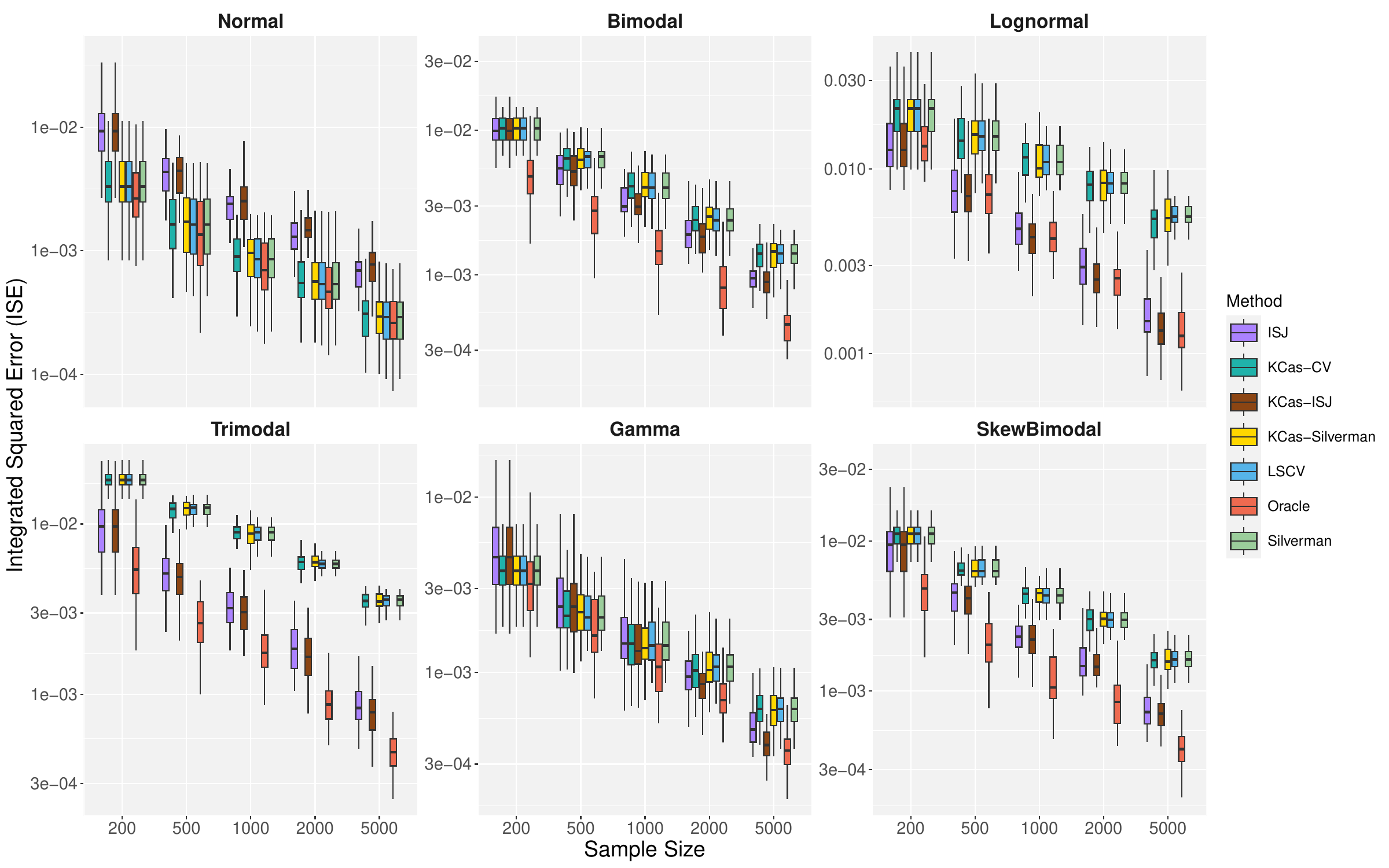}
    \caption{KCas for Kernel Density Estimation. Integrated squared error of KDE bandwidth selection methods across six benchmark densities and five sample sizes. Lower values indicate better performance.}
    \label{fig:kde_amise}
\end{figure}

Figure~\ref{fig:kde_amise} displays the distribution of ISE values across replications for all selectors, stratified by density and sample size. As expected, the oracle AMISE bandwidth achieves the smallest median ISE in every setting and provides a lower bound for the other procedures. Among implementable full-sample selectors, the diffusion-based ISJ bandwidth is consistently the strongest competitor, with median ISEs very close to the oracle across all $n$ and for both simple (normal, gamma) and more complex (bimodal, trimodal, skewed) densities, in line with previous empirical studies \citep{sheather1991reliable,jones1996survey,botev2010kernel}. Silverman's rule-of-thumb performs competitively for the unimodal, approximately Gaussian densities, but becomes noticeably suboptimal for heavier-tailed and highly multimodal shapes, where its normal-reference assumption leads to oversmoothing. The LSCV selector exhibits the largest variability, especially for smaller sample sizes, with wider ISE distributions and occasional outlying bandwidths, reflecting its well-known finite-sample instability \citep{rudemo1982empirical,bowman1984alternative,jones1996survey}.

The KCas variants largely inherit the strengths and weaknesses of their corresponding base selectors while operating on only a small subsample. For the normal and gamma densities, KCas-Silverman attains median ISE ratios on the order of $1.05$–$1.20$ relative to the oracle across the range of $n$, indicating only a modest loss in efficiency despite learning bandwidths from $b=200$ observations. For multimodal and skewed densities, KCas-ISJ is more adaptive and typically dominates KCas-Silverman, with median ISE curves that nearly overlap those of the full-sample ISJ selector. Across all densities, the variability of KCas-ISJ remains close to that of ISJ itself, suggesting that the additional randomness induced by subsampling does not substantially degrade performance. KCas-CV performs reasonably well in smooth unimodal settings, particularly for the normal density, but is generally less competitive for heavier-tailed and multimodal densities and exhibits larger variability, reflecting the weaker finite-sample stability of cross-validation based selectors. Overall, KCas-Silverman and KCas-ISJ consistently outperform KCas-CV, benefiting from more structured statistical foundations and stronger theoretical guarantees.

From the computational perspective, the cost of KCas is dominated by fitting the student model on $b=200$ points, rendering the bandwidth selection cost essentially independent of $n$ once $b$ is fixed. Taken together, these results demonstrate that, in the KDE setting, KCas provides a practical mechanism for transferring the adaptivity of sophisticated selectors such as ISJ to large-sample regimes, while maintaining optimal scaling and substantially reducing the computational burden of hyperparameter tuning.

Simulation results across various settings and methods show that KCas can effectively transfer knowledge of smoothing parameters from the student model on the subsample to the teacher model on the full sample, consistent with our theoretical analysis.

\section{Real Data Analysis}

We apply KCas to real datasets in two settings: smoothing-spline-based nonparametric estimation and deep learning hyperparameter transfer. Specifically, we consider density estimation and nonparametric regression on benchmark datasets, as well as image classification on CIFAR-10.

\subsection{Nonparametric Estimation}

We first evaluate KCas on smoothing-spline-based nonparametric estimation tasks. The density estimation study uses four benchmark datasets, and the nonparametric regression study uses five benchmark datasets. Details of the datasets are provided in Appendix \ref{data}.
Each dataset is randomly split into 80\% training and 20\% testing sets. Features are scaled using a min--max transformation fitted on the training set and then applied to the testing set. All relative metrics in this subsection are computed with respect to full-sample GCV.

\vspace{3pt}

\noindent \textbf{Density estimation.}
We compare KCas with GAM, SUB, ORD, and KDE, in density estimation. Since there is no ground truth for the density function, we evaluate the performance based on the average log-likelihood on the test set, as suggested by previous studies \citep{papamakarios2017masked, gao2022adaptive}. 
We consider the ANOVA decomposition of $\eta$ including all main effects and all two-way interactions. The model terms are selected using the model diagnosis suggested by \citet{gu2004model}. Our results, presented in Table \ref{table:r1}, demonstrate that KCas outperforms all four benchmark methods in terms of log-likelihood across all data sets. Notably, on the ESC and MFCC datasets, KCas even outperforms GCV on the full sample, both in terms of log-likelihood and computational time.
\begin{table}[t]
\centering
\vspace{6pt}
\scalebox{0.9}{
\begin{tabular}{c|c c c c c|c c c c c}
\toprule
 \multicolumn{1}{c}{ }&  \multicolumn{5}{c}{Relative log-likelihood} & \multicolumn{5}{c}{Relative computation time}\\
\toprule
Method& KCas& GAM & SUB & ORD& KDE & KCas& GAM & SUB & ORD& KDE\\
\midrule
CD14  & \textbf{0.9998}   & 0.8095 & 0.7642 & 0.9990 & 0.8342   &0.86 &0.84 & 0.26  &0.48   &7.40\\   
AReM  &\textbf{0.9995}    & 0.9657 & 0.9823& 0.9978 & 0.9568    &0.73 & 0.82 & 0.10 & 0.24  & 1.02\\
ESC  &\textbf{1.0475}    &0.5514  & 1.0369& 1.0191 & 0.2503    &0.53 &  0.87 & 0.42 & 0.44  & 0.69\\
MFCC &\textbf{1.0054}   & 0.9572 & 0.9908& 0.9988 &  0.2528   &0.24 & 0.20  &0.19 &  0.19 &1.73\\
\bottomrule
\end{tabular}
} 
\caption{Relative log-likelihood and relative computational time for different methods in density estimation with real data. GCV is taken as the benchmark method. Higher relative log-likelihood indicates better performance. The highest relative log-likelihood values for each dataset are marked in bold.}
\label{table:r1}
\end{table}

\vspace{3pt}
\noindent \textbf{Nonparametric regression. }
We compare KCas with GAM, SUB, ORD, and SKIP, in nonparametric regression for exponential families. Since we do not know the underlying probability of each data point, as suggested by \citet{wang2018optimal}, we calculate RelMSE over GCV. The main effect and interaction terms are selected using the model diagnosis suggested by \citet{gu2004model}. 
Table \ref{table:r2} shows that KCas outperforms all four benchmark methods in terms of MSE. Although KCas is not the fastest among the methods, it is faster than the full sample estimator in all studies, while obtaining the best performance among the comparison methods.

\begin{table}[t]
\centering
\vspace{6pt}
\scalebox{0.9}{
\begin{tabular}{c|c c c c c|c c c c c}
\toprule
 \multicolumn{1}{c}{ }&  \multicolumn{5}{c}{Relative MSE} & \multicolumn{5}{c}{Relative computation time}\\
\toprule
Method & KCas        &GAM    &SUB    &ORD   &SKIP    &KCas &GAM  &SUB  &ORD  &SKIP\\ \midrule
SUSY   & \bf{0.7963} &0.8333 &1.2185 &0.9385 &0.8252 &0.15 &0.09 &0.11 &0.01 &0.09\\
WFRN   & \bf{1.0434} &1.0535 &1.3604 &1.1071 &1.0827 &0.41 &0.02 &0.21 &0.01 &0.03\\
OCUP & \bf{0.9999} &4.1504 &2.0734 &1.0055 &1.0010 &0.10 &0.16 &0.03 &0.07 &0.13\\
SHILL & \bf{0.9825} & 0.9839 & 1.0218 & 1.0217 & 1.0162 & 0.27 & 0.03 & 0.26 & 0.01 & 0.01 \\
CIFAR-10 & \textbf{1.0040} & 1.0710 & 1.0513 & 1.0264 & 1.0171  & 0.53 & 0.06 & 0.21 & 0.01 & 0.03 \\
\bottomrule
\end{tabular}
}
\caption{Relative MSE and relative computational time for different methods in nonparametric regression with real data. GCV is taken as the benchmark method. Lower relative MSE indicates better performance. The lowest relative MSE values for each dataset are marked in bold.}
\label{table:r2}
\end{table}

\subsection{Image Classification with Deep Learning }\label{sec_real:dl}

We next apply KCas to deep learning hyperparameter transfer on the CIFAR-10 dataset \citep{krizhevsky2009learning}. CIFAR-10 contains $50,000$ training images and $10,000$ test images from ten evenly represented object classes: airplane, automobile, bird, cat, deer, dog, frog, horse, ship, and truck.

\noindent\textbf{Model Structure.}
We consider two student--teacher pairs. MobileNetV2 with 2.2M parameters \citep{sandler2018mobilenetv2} and ResNet18 with 11.7M parameters \citep{he2016deep} serve as student models, while ResNet50 with 25.6M parameters serves as the teacher model. All architectures are adapted to the CIFAR-10 resolution by removing max-pooling and replacing the initial layers with a $3\times3$ convolution. Standard dataset normalization is applied.

\begin{table}[ht]
\centering
\vspace{6pt}
\scalebox{0.9}{
\begin{tabular}{c c c c }
\toprule
& Student 1 & Student 2 & Teacher  \\
\midrule
Model Structure & MobileNetV2 & ResNet18 & ResNet50   \\
Number of Parameters & 2.2M & 11.7M & 25.6M     \\      
\bottomrule
\end{tabular}
}
\caption{Model structure and model size.}
\label{table:dl_model}
\end{table}

\noindent\textbf{KCas Hyperparameter Search.} 
To obtain student configurations for transfer, each student model is tuned by grid search over $12$ hyperparameter combinations: batch sizes ${128,256,512}$, learning rates ${0.05,0.1}$, and weight-decay values ${5\mathrm{e}{-4},1\mathrm{e}{-3}}$. Students are trained for $200$ epochs with SGD using cosine decay and a 5\% warmup. The best configuration is selected using a 10\% holdout validation set. KCas then scales the selected learning rate to the teacher model.

\noindent\textbf{Baselines.}
We compare KCas with three baselines. \textbf{(1) Student model}: the accuracy of the tuned student model. \textbf{(2) Cookbook}: hand-selected hyperparameters commonly recommended for the teacher architecture. \textbf{(3) Retune}: a full grid search for the teacher over $36$ combinations: batch sizes ${256,512,1024}$, learning rates ${0.02,0.05,0.1,0.2}$, and weight-decay values ${1\mathrm{e}{-4},5\mathrm{e}{-4},1\mathrm{e}{-3}}$. This baseline represents the standard but computationally expensive approach to teacher hyperparameter tuning.

\noindent\textbf{Results.}
For each method, we record test accuracy and training time, averaged over three repetitions. 

\begin{figure}[ht]
\centering
\includegraphics[width=0.99\linewidth]{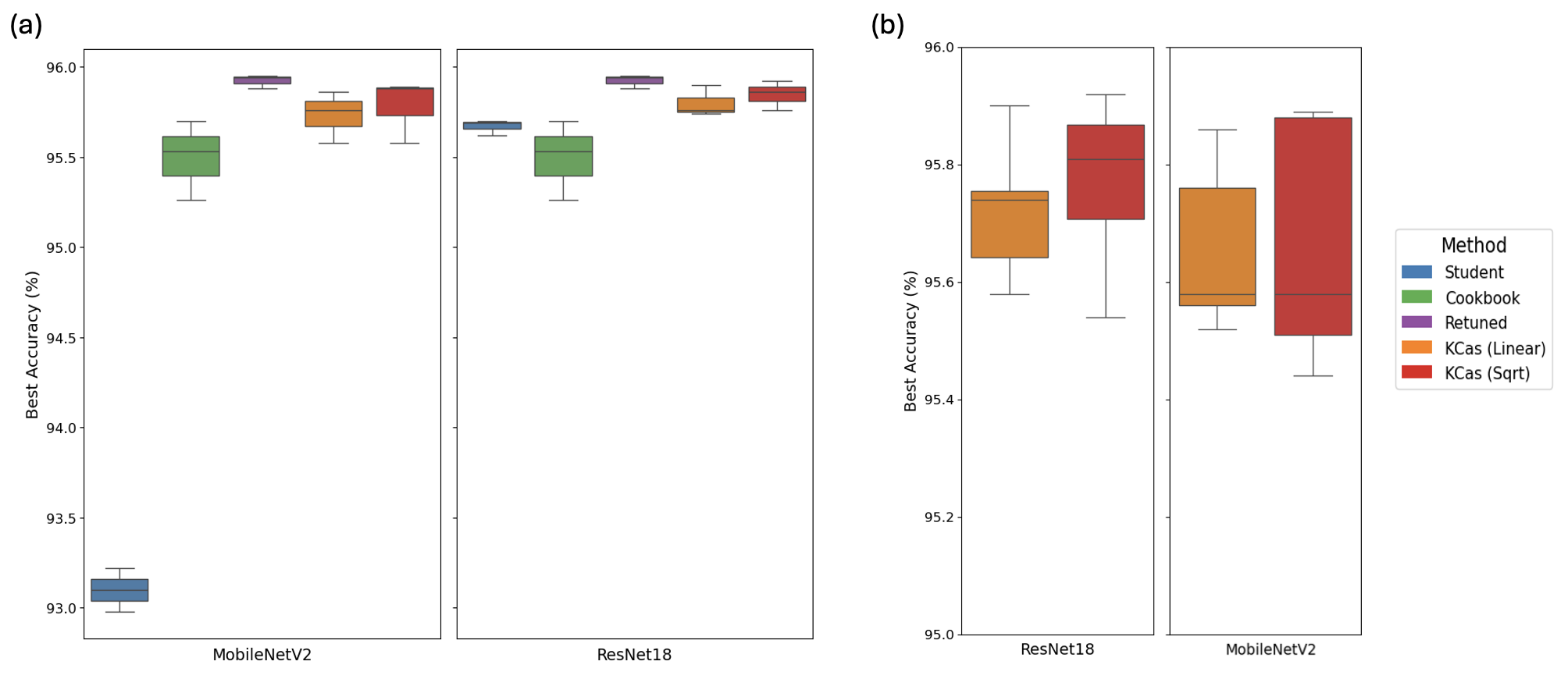}
\caption{KCas for deep learning. We evaluate KCas for transferring learning-rate schedules from a compact student network to a larger teacher network. Two student architectures, MobileNetV2 and ResNet18, are used to guide the training of teacher models. (a) Accuracy comparisons among baselines. (b) Direct comparison of the two KCas scaling rules.
}\label{plot:dl}
\end{figure}

Figure \ref{plot:dl} shows that KCas can effectively transfer learning-rate information from smaller networks to a larger teacher network, achieving teacher performance competitive with full hyperparameter re-tuning. Panel (a) compares KCas with the baselines. Both KCas variants, using linear and square-root scaling, substantially improve over the student models and reach accuracy levels close to those of fully re-tuned teachers for both MobileNetV2 and ResNet18. The cookbook rule provides moderate gains but consistently falls short of the proposed KCas approaches.
An additional pattern in panel (a) is that teacher performance is higher when ResNet18 is used as the student rather than MobileNetV2. This likely reflects the stronger standalone performance of ResNet18, which allows it to provide more informative hyperparameters for transfer. The results suggest that the student model should not be too weak in practice. Although MobileNetV2-based transfer still outperforms the cookbook baseline, a stronger student yields clearer benefits for the teacher.

Panel (b) directly compares the two KCas scaling rules across all batch-size choices, rather than only the best setting. Linear and square-root scaling achieve similar accuracy overall, while the square-root rule shows slightly greater stability across repetitions. Together with panel (a), these results suggest that square-root scaling tends to perform better in this experiment.

\begin{table}[ht]
\centering
\vspace{6pt}
\scalebox{0.9}{
\begin{tabular}{c| c c c c c}
\toprule
Time (h) & Student & KCas & Cookbook & Retune  \\
\midrule
MobileNetV2 & 3.5 & 6.9 & 3.4 & 121.1  \\     
ResNet18 & 18.5 & 21.8 & 3.4 & 121.1  \\ 
\bottomrule
\end{tabular}
}
\caption{Comparison of median computational time (hours) for the image classification task.}
\label{table:dl_time}
\end{table}

Table \ref{table:dl_time} reports the computational time. Compared with full re-tuning of the teacher model, which requires more than $120$ hours, KCas uses only a small fraction of the computational budget. The savings are more pronounced with a smaller student model, since the total cost of KCas depends on the student-training time. In practical scenarios where a trained student model is already available, the additional cost of KCas reduces to training a single teacher model, making it comparable to the cookbook baseline. Although KCas is slightly more expensive than the cookbook approach when the student must be trained from scratch, it consistently achieves better performance. These results illustrate that KCas offers a favorable balance between efficiency and predictive performance for deep learning hyperparameter transfer.


\section{Discussion}
\label{sec:conc}

In this paper, we propose Knowledge Cascade (KCas), a student-to-teacher knowledge transfer framework developed primarily for nonparametric functional estimation in RKHS. In this setting, a student model fitted on a small subsample guides smoothing-parameter selection for the full-sample teacher model, reducing tuning cost while retaining strong statistical performance. Our simulation and real-data results show that KCas compares favorably with several computationally efficient alternatives and can even outperform full-sample GCV in some settings. Beyond the main RKHS setting, we further illustrate the same student-to-teacher transfer principle through kernel density estimation and deep learning hyperparameter transfer. These examples suggest that KCas is not limited to smoothing spline models, but can provide a broader strategy for scalable model development whenever useful information learned from a smaller model can be transferred to a larger one. More broadly, KCas offers a new perspective on knowledge distillation by showing that small models can sometimes guide, rather than merely approximate, larger models. It remains an open question whether such knowledge cascades can be constructed without the support of asymptotic or empirical scaling rules.

\noindent \textbf{When and why to use KCas.}
KCas is most useful when a smaller student model can learn information that remains informative for developing a larger teacher model. The key requirement is not that the student approximates the teacher perfectly, but that it captures transferable information that can reduce the cost of training or tuning the teacher. Such information may include smoothing parameters, bandwidths, training configurations, or other structural features that are expensive to learn directly at full scale. When this information follows a suitable asymptotic or empirical scaling relationship, KCas provides a way to extract it from a low-cost student model and propagate it to the teacher.

KCas is not intended to replace the teacher with a smaller approximation. Instead, the student serves as a computationally efficient guide for building the teacher. In this paper, we instantiate this idea mainly through hyperparameter transfer. For example, in smoothing spline ANOVA models, KCas learns smoothing parameters from a subsample-based student model and transfers them to the full-sample teacher model through an asymptotic scaling rule. Thus, KCas avoids repeated full-scale tuning while still fitting the final model on all observations. More broadly, the same principle suggests a general strategy for reducing the cost of teacher-model development whenever useful information can be learned cheaply from a smaller model and reliably transferred to a larger one.

\noindent \textbf{KCas can outperform the full sample estimator.}
In multiple scenarios in our simulation and real data analysis, we have observed an interesting phenomenon that KCas, based on a random subsample, can outperform the full sample GCV estimator. One potential explanation is that the statistical theory underlying KCas helps make hyperparameter transfer more efficient. Although this finding appears to conflict with the traditional intuition, 
recent studies \citep{nakkiran2021deep,guo2022deepcore,yang2022dataset,sorscher2022beyond,gadre2024datacomp}, have observed similar phenomena that models trained on subsamples sometimes can achieve better performance in the context of deep models. 
Further, a related paper \citep{kolossov2023towards} observed such a phenomenon for simpler settings under empirical risk minimization, and they aim to develop some theoretical justification for it. Although they work under different model settings, it is an interesting future direction to explore this phenomenon theoretically and empirically in nonparametric estimation or deep learning settings.

\noindent \textbf{Connection to broader statistical methodology.}
The KCas principle may also be relevant to other statistical methods whose effective complexity changes with sample size. Examples include basis-expansion methods, wavelet estimators, series regression, and other multiresolution procedures. In these settings, quantities such as the number of basis terms, resolution levels, thresholding parameters, or regularization strengths often need to be adjusted as the sample size increases. This makes them natural candidates for KCas-type extensions, since a student model fitted at a smaller scale may help estimate problem-specific quantities that guide the corresponding full-scale estimator. Establishing suitable scaling rules and theoretical guarantees for these broader classes of methods remains an important direction for future work.

\noindent \textbf{Code availability.} The source code is available at: \url{https://github.com/LuyangFang/KCas}.

\acks{This work was partially supported by the U.S. National Science Foundation under grants  DMS-1903226, DMS-1925066, DMS-2124493, DMS-2311297, DMS-2319279, DMS-2318809, and by the U.S. National Institutes of Health under grant R01GM152814. The authors declare no competing interests.
}


\newpage
\appendix


\numberwithin{equation}{section} 

\section{Existence of the Minimizer}\label{Exist}

The following theorem guarantees the existence of the minimizer of (\ref{PLik}) in RKHS.
\begin{appendixtheorem}[Existence of the minimizer,  \citet{wahba1990spline}]
Suppose $L(\eta)$ is a continuous and convex functional in a Hilbert space $\mathcal{H}$ and $J(\eta)$ is a square (semi) norm in $\mathcal{H}$ with a null space $\mathcal{N}_{J}$, of finite dimension. If $L(\eta)$ has a unique minimizer in $\mathcal{N}_{J}$, then $L(\eta)+ \frac{\lambda}{2}J(\eta)$ has a minimizer in $\mathcal{H}$.
\label{thm:existence}
\end{appendixtheorem}

When $L(\eta)$ is the negative log-likelihood function, it is usually convex in $\eta$. The quadratic functional $J(\eta)$ is also convex \citep{gu1993smoothing}.  A minimizer of $L(\eta)$ is unique in $\mathcal{N}_{J}$ if the convexity is strict in it, which is often the case. Thus, the solution for Equation (\ref{PLik}) exists in most cases.

\section{Minimizer of the Penalized Loss Functional}\label{Minimizer}

Consider exponential family distributions with densities of the form
\begin{equation}
    f(y \mid x)=\exp \{\frac{y \vartheta(x)-b(\vartheta(x))}{a(\phi)}+c(y, \phi)\},
\end{equation}
where $a>0, b$, and $c$ are known functions, $\vartheta(x)$ is the canonical parameter dependent on a covariate $x$, and $\phi$ is either known or considered as a nuisance parameter that is independent of $x$. 
Fixing the smoothing parameters, the penalized likelihood functional (\ref{eq:PL_reg}) is strictly convex in $\eta$. Thus, given the current $\Tilde{\eta}$, the Newton iteration can be used to update $\Tilde{\eta}$ by the minimizer of the penalized weighted least square functional 

\begin{equation}
    \frac{1}{n} \left( \tilde{\bsY} - S \mathbf{d} -Q \mathbf{c} \right)^T W \left( \tilde{\bsY} - S \mathbf{d} -Q \mathbf{c} \right) + \frac{\lambda}{2} \mathbf{c}^{T} Q \mathbf{c},
\end{equation}
where $\tilde{Y}_i=\tilde{\eta}\left(x_i\right)- \frac{\tilde{u}_i} {\tilde{w}_i}$,
$\tilde{w}_i=\ddot{b}\left(\tilde{\eta}\left(x_i\right)\right), \text{and } \tilde{u}_i=-Y_i+\dot{b}\left(\tilde{\eta}\left(x_i\right)\right)$. 
Here $\tilde{w}_i$ is the $i$th diagonal element of the matrix $W$, $\tilde{Y}_i$ is the $i$th element of $\tilde{\bsY}$, $\dot{b}$ and $\ddot{b}$ are the first and second derivatives of the function $b$.


\section{Example of the Tensor Product Cubic Spline on $[0,1]^2$}
For the univariate $\eta$ on $\mathcal{X}$, the most popular choice of the smoothness penalty $J(\eta)$ is
   \begin{equation}
   J(\eta,\eta) = \int_{\mathcal{X}} \left(\eta^{(m)}\right)^2 \mathrm{~d} x,
   \end{equation}
   where $\eta^{(m)}=\mathrm{d}^m \eta / \mathrm{d} x^m$. A cubic estimator of the minimizer of (\ref{PLik}) is obtained by setting $m=2$. This idea can be extended to multivariate settings. 
   Consider the ANOVA decomposition (\ref{eq:decomposition}). $\mathcal{H}$ can be decomposed into the space of constants, the spaces of main effects, and the corresponding spaces of interaction terms lying in the tensor product space of the interacting main-effect spaces. For the two-dimensional problem, one has the following space decomposition in each variable (\citet{gu2013smoothing}, section 2.3):
   \begin{equation}
       \begin{aligned}
        \left\{\eta: \eta^{(2)} \in \mathcal{L}_2[0,1]\right\}= & \{\eta: \eta \propto 1\} \oplus\left\{\eta: \eta \propto k_1\right\} \\
        & \oplus\left\{\eta: \int_0^1 \eta \mathrm{~d} x=\int_0^1 \eta^{(1)} \mathrm{d} x=0, \eta^{(2)} \in \mathcal{L}_2[0,1]\right\} \\
        = & \mathcal{H}_{00} \oplus \mathcal{H}_{01} \oplus \mathcal{H}_1,
        \end{aligned}
   \end{equation}
   where $k_1(x)=x-0.5$. The space of constant terms is $\mathcal{H}_{00\langle1\rangle} \otimes \mathcal{H}_{00\langle2\rangle}$; the space of main effects is spanned by $\mathcal{H}_{00\langle1\rangle} \otimes\left(\mathcal{H}_{01\langle2\rangle} \oplus \mathcal{H}_{1\langle2\rangle}\right)$ and 
   $\mathcal{H}_{00\langle 2\rangle} \otimes\left(\mathcal{H}_{01\langle 1\rangle} \oplus \mathcal{H}_{1\langle1\rangle}\right)$; 
   and the subspace $\left(\mathcal{H}_{01\langle 1\rangle} \oplus \mathcal{H}_{1\langle 1\rangle}\right) \otimes$ $\left(\mathcal{H}_{01\langle2\rangle} \oplus \mathcal{H}_{1\langle2\rangle}\right)$ spans the space of interactions. Let $\mathcal{H}_{\nu, \mu}=\mathcal{H}_{\nu\langle 1\rangle} \otimes \mathcal{H}_{\mu\langle 2\rangle}$ for $\nu, \mu=00,01,1$, with inner products $(\eta, \eta)_{\nu, \mu}$ and reproducing kernels $R_{\nu, \mu}=R_{\nu\langle 1\rangle} R_{\mu\langle 2\rangle}$, using the tensor product cubic spline, we have 
   \begin{equation}
       \begin{aligned}
        J(\eta, \eta)= & \theta_{1,00}^{-1}(\eta, \eta)_{1,00}+\theta_{00,1}^{-1}(\eta, \eta)_{00,1} \\
        & +\theta_{1,01}^{-1}(\eta, \eta)_{1,01}+\theta_{01,1}^{-1}(\eta, \eta)_{01,1}+\theta_{1,1}^{-1}(\eta, \eta)_{1,1} ,
        \end{aligned}
   \end{equation}
   and the null space of $J(\eta, \eta)$ is 
   \begin{equation}
       \mathcal{N}_J=\mathcal{H}_{00,00} \oplus \mathcal{H}_{01,00} \oplus \mathcal{H}_{00,01} \oplus \mathcal{H}_{01,01} .
   \end{equation}

\section{Regularity Conditions}\label{condi}



We define the quadratic functional representing the mean square error of the estimator $\hat{\eta}$ in estimating the target function $\eta_0$ on the domain $\mathcal{X}$ as
\begin{equation*}
    V\left(\hat{\eta}-\eta_0 \right)=\int_{\mathcal{X}}\left\{\hat{\eta}-\eta_0(x)\right\}^{2} f(x) \mathrm{d} x,
\end{equation*}
where $f(x)$ is the marginal density of $x$.
We now state some regularity conditions for Theorem \ref{thm:converge_den} and Theorem \ref{thm:converge}.

\begin{appcondition}\label{cond:continuous}
The functional V is completely continuous with respect to J.
\end{appcondition}
When condition \ref{cond:continuous} is satisfied, that is, $V$ is completely continuous with respect to $J$ and hence to $V+J$, there exist eigenvalues $\lambda_{\nu}$ and the corresponding eigenfunctions $\psi_{\nu}$ such that
\begin{equation*}
\begin{aligned}
    &V\left(\psi_{\nu}, \psi_{\mu}\right)=\lambda_{\nu} \delta_{\nu, \mu}, \text{ and } \\
    &(V+J)\left(\psi_{\nu}, \psi_{\mu}\right)=\delta_{\nu, \mu},
\end{aligned}
\end{equation*}
where $\delta_{\nu, \mu}$ is the Kronecker delta and $1 \geq \lambda_{\nu} \downarrow 0$; see \cite{weinberger1974variational},
\cite{silverman1982est}. 

Write $\phi_{\nu}=\lambda_{\nu}^{-\frac{1}{2}} \psi_{\nu}$. It follows that
\begin{equation*}
\begin{aligned}
    V\left(\phi_{\nu}, \phi_{\mu}\right)&=\delta_{\nu, \mu},\\
    J\left(\phi_{\nu}, \phi_{\mu}\right)&=\rho_{\nu} \delta_{\nu, \mu},
\end{aligned}
\end{equation*}
where $0 \leq \rho_{\nu}=\lambda_{\nu}^{-1}-1$. We refer to $\rho_{\nu}$ as the eigenvalues of $J$ with respect to $V$ and to $\phi_{\nu}$ as the associated eigenfunctions. A Fourier series expansion of $\eta_0$ satisfying $J(\eta_0)<\infty$ is $\eta_0=\sum_{\nu} \eta_{\nu,0} \phi_{\nu}$, where $\eta_{\nu,0}=V\left(\eta_0, \phi_{\nu}\right)$ are the Fourier coefficients. 

\begin{appcondition}\label{cond:summation}
    $\sum_{\nu} \rho_{\nu}^{p} \eta_{\nu, 0}^{2}<\infty$ for some $p \in[1,2]$.
\end{appcondition}

\begin{appcondition}\label{cond:eigenvalue}
For $\nu$ sufficiently large and some $\beta>0$, the eigenvalues $\rho_{\nu}$ of $J$ with respect to $V$ satisfy $\rho_{\nu}>\beta \nu^{r}$, where $r>1$.
\end{appcondition}

\begin{appcondition} \label{cond:equi_den}
For $\eta$ in a convex set $B_0$ around $\eta_0$ containing $\hat{\eta}$ and $\tilde{\eta}$, where $\Tilde{\eta}$ is a linear approximation of $\hat{\eta}$, $c_1 V(f) \leq V_\eta(f)$ holds uniformly for some $c_1>0$.
\end{appcondition}

\begin{appcondition} \label{cond:bound}
$\operatorname{Var}\left[\phi_{\nu}(X) \phi_{\mu}(X) w\left(\eta(X), Y\right)\right] \leq c_{3}$ for some $c_{3}<\infty$, $\forall \nu, \mu$.
\end{appcondition}
Condition \ref{cond:bound} requires a uniform bound for the fourth moments of $\phi_{\nu}(X)$.

\begin{appcondition}\label{cond:equi}
Let $w\left(\eta; Y\right) = \frac{d^2 l}{d \eta^2}$, where $l\left( \eta; Y \right)$ is the minus log likelihood of $\eta$ with observations $Y$.
For $\tilde{\eta}$ in a convex set $B_{0}$ around $\eta_0$ containing $\hat{\eta},\ c_{1} w\left(\eta_0(x) ; Y\right) 
\leq w(\tilde{\eta}(x) ; Y) 
\leq c_{2} w\left(\eta_0(x) ; Y\right)$ holds uniformly for some $0<c_{1}<c_{2}<\infty, \forall x \in \mathcal{X}, \forall Y$.
\end{appcondition}
Condition \ref{cond:equi} requires the equivalence of the information in $B_0$.
\section{Proofs of Main Results}\label{proof}

\begin{proof}[Proof of Theorem \ref{thm:converge_den}]
We start by summarizing the notations used in the theorem. $\eta_0$ is the true function.
$\hat{\eta}$ is the estimation based on $\lambda_{\mathrm{KCas}}^{\mathrm{full}}(n;b) $. 
For the sake of simplicity, we write $r=2m$ in the proof. Recall that 
\begin{equation*}
     \lambda_{\mathrm{KCas}}^{\mathrm{full}}(n;b) 
     = \lambda_{\mathrm{GCV}}^{\mathrm{sub}}(b)
     (\frac{n}{b})^{-\frac{r}{rp+1}}.
\end{equation*}
It suffices to show that as $n \rightarrow \infty$, $$\lambda_{\mathrm{KCas}}^{\mathrm{full}}(n;b)  \rightarrow 0, \text{ and }$$
$$n (\lambda_{\mathrm{KCas}}^{\mathrm{full}}(n;b) )^{\frac{1}{r}} \rightarrow \infty.$$
Since $\lambda_{\mathrm{GCV}}^{\mathrm{sub}}(b) \rightarrow 0$ and $(\frac{n}{b})^{-\frac{r}{rp+1}}<1$, we have 
\begin{equation*}
    \lambda_{\mathrm{KCas}}^{\mathrm{full}}(n;b) 
     = \lambda_{\mathrm{GCV}}^{\mathrm{sub}}(b)
     (\frac{n}{b})^{-\frac{r}{rp+1}} \rightarrow 0.
\end{equation*}
Also, since $rp>1$, we have
\begin{align}
n (\lambda_{\mathrm{KCas}}^{\mathrm{full}}(n;b) )^{\frac{1}{r}} 
&=  n (\lambda_{\mathrm{GCV}}^{\mathrm{sub}}(b))^{\frac{1}{r}}
     (\frac{n}{b})^{-\frac{1}{rp+1}} \nonumber\\
&= n^{\frac{rp}{rp+1}}  b^{\frac{1}{rp+1}} (\lambda_{\mathrm{GCV}}^{\mathrm{sub}}(b))^{\frac{1}{r}} \nonumber\\
&\geq  b^{\frac{rp}{rp+1}}  b^{\frac{1}{rp+1}} (\lambda_{\mathrm{GCV}}^{\mathrm{sub}}(b))^{\frac{1}{r}} \nonumber\\
& = b (\lambda_{\mathrm{GCV}}^{\mathrm{sub}}(b))^{\frac{1}{r}} \label{eq: thm1proof}
\rightarrow \infty.
\end{align}
Therefore, $n (\lambda_{\mathrm{KCas}}^{\mathrm{full}}(n;b) )^{\frac{1}{r}} \rightarrow \infty$. According to Chapter 9 of 
\cite{gu2013smoothing}, we have
\begin{equation*} 
    (V+ \lambda_{\mathrm{KCas}}^{\mathrm{full}}(n;b)  J)\left(\hat{\eta}-\eta_0\right)=O_{p}\left(n^{-1}  \lambda_{\mathrm{KCas}}^{\mathrm{full}}(n;b) ^{-\frac{1}{r}}+ \lambda_{\mathrm{KCas}}^{\mathrm{full}}(n;b) ^{p}\right).
\end{equation*} 
\end{proof}

\begin{proof}[Proof of Theorem \ref{thm:converge}]
Analogous to Theorem \ref{thm:converge_den}, it suffices to show that as $n \rightarrow \infty$, $$\lambda_{\mathrm{KCas}}^{\mathrm{full}}(n;b)  \rightarrow 0, \text{ and }$$
$$n (\lambda_{\mathrm{KCas}}^{\mathrm{full}}(n;b) )^{\frac{2}{r}} \rightarrow \infty.$$
Since $\lambda_{\mathrm{GCV}}^{\mathrm{sub}}(b) \rightarrow 0$ and $(\frac{n}{b})^{-\frac{r}{rp+1}}<1$, we have 
\begin{equation*}
    \lambda_{\mathrm{KCas}}^{\mathrm{full}}(n;b) 
     = \lambda_{\mathrm{GCV}}^{\mathrm{sub}}(b)
     (\frac{n}{b})^{-\frac{r}{rp+1}} \rightarrow 0.
\end{equation*}
Also, since $rp>1$, analogous to Equation (\ref{eq: thm1proof}), we have
\begin{equation*}
\begin{aligned}
n (\lambda_{\mathrm{KCas}}^{\mathrm{full}}(n;b) )^{\frac{2}{r}}
& \geq b (\lambda_{\mathrm{GCV}}^{\mathrm{sub}}(b))^{\frac{2}{r}}
\rightarrow \infty.
\end{aligned}
\end{equation*}
Therefore, $n (\lambda_{\mathrm{KCas}}^{\mathrm{full}}(n;b) )^{\frac{2}{r}} \rightarrow \infty$. According to Chapter 9 of 
\cite{gu2013smoothing}, we have
\begin{equation*} 
    (V+ \lambda_{\mathrm{KCas}}^{\mathrm{full}}(n;b)  J)\left(\hat{\eta}-\eta_0\right)=O_{p}\left(n^{-1}  \lambda_{\mathrm{KCas}}^{\mathrm{full}}(n;b) ^{-\frac{1}{r}}+ \lambda_{\mathrm{KCas}}^{\mathrm{full}}(n;b) ^{p}\right).
\end{equation*} 
\end{proof}
Note that it has been proved rigorously that the optimal smoothing parameter $\lambda(b)$ has the form $C b^{-\frac{r}{rp+1}}$ 
under some exponential regression problems, such as regression with Gaussian-type responses and periodic splines 
\citep{wahba1977practical, wahba1985comparison,craven1978smooth}. In such cases, with the fact that $rp>1$, as $b \rightarrow \infty$,
\begin{equation*}
\begin{aligned}
    \lambda(b)&= C b^{-\frac{r}{rp+1}} \rightarrow 0, \text{ and} \\
    b \lambda(b)^{\frac{2}{r}} &= b C^{\frac{2}{r}} b^{-\frac{2}{rp+1}} \\
    &= C^{\frac{2}{r}}b^{\frac{rp-1}{rp+1}} \rightarrow \infty,
\end{aligned}
\end{equation*}
that is, $\lambda(b) \rightarrow 0$ and $b (\lambda(b))^{\frac{2}{r}} \rightarrow \infty$ is naturally satisfied. In some cases, such as the density estimation problems and more general exponential-family settings, we impose the stated assumptions on $\lambda_{\mathrm{GCV}}^{\mathrm{sub}}(b)$; the numerical results support their validity.

We replace $\lambda(b)$ with $\lambda_{\mathrm{GCV}}^{\mathrm{sub}}(b)$ chosen by GCV since it is infeasible to determine $\lambda(b)$ with the unknown function $\eta_0$. Theoretical results \citep{li1986asymptotic,craven1978smooth} have shown that $\lambda_{\mathrm{GCV}}^{\mathrm{sub}}(b)$ is a good estimator of $\lambda(b)$, with $\frac{L(\lambda_{\mathrm{GCV}}^{\mathrm{sub}}(b))}{L(\lambda(b))}=1+o_p(1)$.
Thus, it is natural to extend the assumption $\lambda_{\mathrm{GCV}}^{\mathrm{sub}}(b) \rightarrow 0$ and $b (\lambda_{\mathrm{GCV}}^{\mathrm{sub}}(b))^{\frac{2}{r}} \rightarrow \infty$ to the general regression problems with responses from exponential families.
The numerical results also support this assumption.




\section{Simulation Details}

\subsection{Nonparametric regression}\label{Appx:sim_gss}

\textbf{Scenario 1:} Let 
\begin{equation*}
\begin{aligned}
\eta_{m 1}(x) =
  \sum_{i=1}^3 g_1\left(x_{\langle i \rangle}\right) + g_2\left(x_{\langle 1 \rangle}, x_{\langle 2 \rangle}\right) + g_2\left(x_{\langle 1 \rangle}, x_{\langle 3 \rangle}\right) + g_3\left(x_{\langle 1 \rangle}, x_{\langle 2 \rangle}, x_{\langle 3 \rangle}\right),
\end{aligned}
\end{equation*}

\noindent\textbf{Scenario 2:} Let 
\begin{equation*}
\begin{aligned}
\eta_{m 2}(x) &=
  \sum_{i=1}^3 i g_1\left(x_{\langle i \rangle}\right) + \sum_{i=4}^6 i g_{5}\left(x_{\langle i \rangle}\right) + \sum_{i=7}^9 g_{4}\left(x_{\langle i \rangle}\right) + \\ & \sum_{i=1}^3 \sum_{j>i}^4 3i g_2\left(x_{\langle i \rangle}, x_{\langle j \rangle}\right) + 6 g_2\left(x_{\langle 5 \rangle}, x_{\langle 6 \rangle}\right) + 8 g_6\left(x_{\langle 7 \rangle}, x_{\langle 8 \rangle}\right) + 10 g_3\left(x_{\langle 1 \rangle}, x_{\langle 2 \rangle}, x_{\langle 3 \rangle}\right),
\end{aligned}
\end{equation*}
where \\ 
$g_1(x)=10^6 x^{11}(1-x)^6$;\\
$g_2\left(x_{\langle 1 \rangle}, x_{\langle 2 \rangle}\right)=\exp \left(3 x_{\langle 1 \rangle} x_{\langle 2 \rangle}\right)$; \\ 
$g_3\left(x_{\langle 1 \rangle}, x_{\langle 2 \rangle}, x_{\langle 3 \rangle}\right)=$ $\frac{15 \sin \left(2 \pi x_{\langle 1 \rangle}\right)}{2-\sin \left(2 \pi x_{\langle 2 \rangle} x_{\langle 3 \rangle}\right)}$; \\
$g_4(x)=10^4 x^{3}(1-x)^{10}$;\\ 
$g_5(x)=15x sin(15x)$; \\ 
$g6(x) =
 \frac{a p_1}{\pi \sigma_{1} \sigma_{2}} \exp \left\{-\frac{\left(x_{\langle 1\rangle}-0.2\right)^{2}}{\sigma_{1}^{2}}-\frac{\left(x_{\langle 2\rangle}-0.3\right)^{2}}{\sigma_{2}^{2}}\right\} 
+\frac{a p_2}{\pi \sigma_{1} \sigma_{2}} \exp \left\{-\frac{\left(x_{\langle 1\rangle}-0.7\right)^{2}}{\sigma_{1}^{2}}-\frac{\left(x_{\langle 2\rangle}-0.8\right)^{2}}{\sigma_{2}^{2}}\right\}-b,$ with $\sigma_{1}=0.3$, $\sigma_{2}=0.4$, $p_1=0.625$, $p_2=0.375$, and $a=b=4.2$. \\

\section{Datasets}\label{data}
\subsection{Datasets for Density Estimation}

\begin{itemize}
    \item \textit{CD14}: Transcriptions in CD14 single cells. The data contains the abundance information of 13 proteins in 2,096 cells. The data set is available from \cite{stoeckius2017large}.
    \item \textit{AReM}: 
    Activity Recognition system based on Multisensor data fusion Data Set. The dimension is 6 and the sample size is 42,240. 
    The time-domain features including 3 mean values and 3 standard deviations were collected from the multisensor system during a period of time. The data set is available at UCI Machine Learning Repository \citep{lichman2013uci}.
    \item \textit{ESC}: Embryonic Stem Cell from Mouse \citep{ouyang2009chip}. The data concerns mouse embryonic stem cell gene expression and transcription factor association strength.
    The 4 features that describe the scores of TFAS with KLF4, NANOG, OCT4, and SOX2 of 1,027 genes are used for density estimation.
    The data set is available at CRAN in \textit{gss} package.
    \item \textit{MFCC}: Anuran Calls (MFCCs) Data Set. The data is extracted from syllables of anuran (frogs) calls, including 22 variables with a sample size of 7,195. The data set is available at UCI Machine Learning Repository \citep{lichman2013uci}. 
\end{itemize}


All the continuous variables in these datasets are scaled through a min-max normalization.

\subsection{Datasets for Nonparametric Regression}

\begin{itemize}
    \item \textit{SUSY}: Supersymmetric Dataset \citep{baldi2014searching}.
    The dataset contains one response and 18 kinematic features $x_{\langle 1\rangle},\dots,x_{\langle 18\rangle}$.
    The full sample size is 5,000,000, and
    about $54.24\%$ of the responses in
    the data are from the background process. 
    We consider the full sample GCV as the gold standard, but it is not affordable to compute GCV on the full sample. Therefore, we randomly select a subsample of size $20,000$ and consider this sample to be the ``full sample'' to compute the GCV, and only use this part of data to conduct all the following analysis in the paper.
    The data set is available at UCI Machine Learning Repository \citep{lichman2013uci}.
    
    \item \textit{WFRN:} Wall-Following Robot Navigation Dataset \citep{freire2009short}. The data is a robot navigating through the room following the wall using 24 ultrasound sensors with a sample size of 19,735. The data set is available at UCI Machine Learning Repository \citep{lichman2013uci}.


    \item \textit{OCUP}: Room Occupancy Estimation Dataset \citep{singh2018machine}. The experimental testbed for occupancy estimation was deployed in a room. The setup consisted of 7 sensor nodes and one edge node in a star configuration with the sensor nodes transmitting data to the edge every 30s using wireless transceivers. The dataset contains one response and 15 features. The full sample size is 10,129.
    The data set is available at UCI Machine Learning Repository \citep{lichman2013uci}.

    \item \textit{SHILL}: 
    Shill Bidding Dataset \citep{alzahrani2018scraping}. This is a dataset with a large number of eBay auctions of a popular product. The dataset contains one response and 12 features. The full sample size is 6,321. The data set is available at UCI Machine Learning Repository \citep{lichman2013uci}.

    \item \textit{CIFAR-10}:
     The CIFAR-10 dataset \citep{krizhevsky2009learning} consists of a training set of $50,000$ examples and a test set of $10,000$ examples. Each example in the dataset is a $32\times32$ color image, spanning 10 different classes of objects such as animals and vehicles. These classes include airplanes, cars, birds, cats, deer, dogs, frogs, horses, ships, and trucks, each equally represented in the dataset.
     
     For this dataset, we use a pre-trained convolutional neural network to extract the features from the raw images first. This neural network model consists of two main components: a convolutional layer block and a fully-connected layer block. The convolutional block comprises two sets of convolutional layers with batch normalization and max pooling layers for feature extraction from input images. The fully-connected block contains four linear layers with ReLU activations. The third linear layer with 20 nodes serves as a feature extraction layer, providing a compressed representation of the input features for downstream tasks.

     For this experiment, we convert CIFAR-10 into a binary logistic regression task by setting $Y=1$ for the car class and $Y=0$ for all other classes. The MSE is computed on the test set between the binary responses and the fitted conditional probabilities, and the relative MSE is normalized by the MSE of the full-sample estimator.
    
\end{itemize}

All the continuous predictors in these datasets are scaled through a min-max normalization. 
For all the datasets, we apply the cosine diagnostics \citep{gu2013smoothing} first for the identifiability
and the practical significance of the fitted terms, in order to avoid overfitting and overinterpreting.





\bibliography{ref,supp}

@article{meng2020more,
  title={More efficient approximation of smoothing splines via space-filling basis selection},
  author={Meng, Cheng and Zhang, Xinlian and Zhang, Jingyi and Zhong, Wenxuan and Ma, Ping},
  journal={Biometrika},
  volume={107},
  number={3},
  pages={723--735},
  year={2020},
  publisher={Oxford University Press}
}

@article{hoffmann2022training,
  title={Training compute-optimal Large Language Models},
  author={Hoffmann, Jordan and Borgeaud, Sebastian and Mensch, Arthur and Buchatskaya, Elena and Cai, Trevor and Rutherford, Eliza and Casas, DDL and Hendricks, Lisa Anne and Welbl, Johannes and Clark, Aidan and others},
  journal={arXiv preprint arXiv:2203.15556},
  volume={10},
  year={2022}
}

@article{grattafiori2024llama,
  title={The Llama 3 Herd of Models},
  author={Grattafiori, Aaron and Dubey, Abhimanyu and Jauhri, Abhinav and others},
  journal={arXiv preprint arXiv:2407.21783},
  year={2024}
}

@article{openai2023gpt4,
  title={{GPT}-4 Technical Report},
  author={Achiam, Josh and Adler, Steven and Agarwal, Sandhini and Ahmad, Lama and Akkaya, Ilge and Aleman, Florencia Leoni and Almeida, Diogo and Altenschmidt, Janko and Altman, Sam and Anadkat, Shyamal and others},
  journal = {arXiv preprint arXiv:2303.08774},
  year    = {2023},
  doi     = {10.48550/arXiv.2303.08774},
  url     = {https://arxiv.org/abs/2303.08774}
}

@article{fang2025spot,
  title={Spot: an active learning algorithm for efficient deep neural network training},
  author={Fang, Luyang and Meng, Cheng and Zhao, Lin and Wang, Tao and Liu, Tianming and Zhong, Wenxuan and Ma, Ping},
  journal={Big Data Mining and Analytics},
  volume={8},
  number={5},
  pages={1060--1074},
  year={2025},
  publisher={TUP}
}

@article{fang2026knowledge,
  title={Knowledge distillation and dataset distillation of large language models: Emerging trends, challenges, and future directions},
  author={Fang, Luyang and Yu, Xiaowei and Cai, Jiazhang and Chen, Yongkai and Wu, Shushan and Liu, Zhengliang and Yang, Zhenyuan and Lu, Haoran and Gong, Xilin and Liu, Yufang and others},
  journal={Artificial Intelligence Review},
  volume={59},
  number={1},
  pages={17},
  year={2026},
  publisher={Springer}
}

@article{gou2021knowledge,
  title={Knowledge distillation: A survey},
  author={Gou, Jianping and Yu, Baosheng and Maybank, Stephen J and Tao, Dacheng},
  journal={International Journal of Computer Vision},
  volume={129},
  number={6},
  pages={1789--1819},
  year={2021},
  publisher={Springer}
}

@inproceedings{zhang2019your,
  title={Be your own teacher: Improve the performance of convolutional neural networks via self distillation},
  author={Zhang, Linfeng and Song, Jiebo and Gao, Anni and Chen, Jingwei and Bao, Chenglong and Ma, Kaisheng},
  booktitle={Proceedings of the IEEE/CVF international conference on computer vision},
  pages={3713--3722},
  year={2019}
}

@inproceedings{fang2024bayesian,
  title={Bayesian knowledge distillation: a Bayesian perspective of distillation with uncertainty quantification},
  author={Fang, Luyang and Chen, Yongkai and Zhong, Wenxuan and Ma, Ping},
  booktitle={Forty-first International Conference on Machine Learning},
  year={2024}
}

@article{ma2026generalizable,
  title={A generalizable pathology foundation model using a unified knowledge distillation pretraining framework},
  author={Ma, Jiabo and Guo, Zhengrui and Zhou, Fengtao and Wang, Yihui and Xu, Yingxue and Li, Jinbang and Yan, Fang and Cai, Yu and Zhu, Zhengjie and Jin, Cheng and others},
  journal={Nature Biomedical Engineering},
  volume={10},
  number={3},
  pages={545--564},
  year={2026},
  publisher={Nature Publishing Group UK London}
}

@inproceedings{dehghani2023scaling,
  title={Scaling vision transformers to 22 billion parameters},
  author={Dehghani, Mostafa and Djolonga, Josip and Mustafa, Basil and Padlewski, Piotr and Heek, Jonathan and Gilmer, Justin and Steiner, Andreas Peter and Caron, Mathilde and Geirhos, Robert and Alabdulmohsin, Ibrahim and others},
  booktitle={International Conference on Machine Learning},
  pages={7480--7512},
  year={2023},
  organization={PMLR}
}

@article{goyal2017accurate,
  title={Accurate, large minibatch sgd: Training imagenet in 1 hour},
  author={Goyal, Priya and Doll{\'a}r, Piotr and Girshick, Ross and Noordhuis, Pieter and Wesolowski, Lukasz and Kyrola, Aapo and Tulloch, Andrew and Jia, Yangqing and He, Kaiming},
  journal={arXiv preprint arXiv:1706.02677},
  year={2017}
}

@inproceedings{xie2020self,
  title={Self-training with noisy student improves imagenet classification},
  author={Xie, Qizhe and Luong, Minh-Thang and Hovy, Eduard and Le, Quoc V},
  booktitle={Proceedings of the IEEE/CVF Conference on Computer Vision and Pattern Recognition},
  pages={10687--10698},
  year={2020}
}

@article{gu2013nonparametric,
  title={Nonparametric density estimation in high-dimensions},
  author={Gu, Chong and Jeon, Yongho and Lin, Yi},
  journal={Statistica Sinica},
  pages={1131--1153},
  year={2013},
  publisher={JSTOR}
}

@article{gu1993smoothing,
  title={Smoothing spline density estimation: Theory},
  author={Gu, Chong and Qiu, Chunfu},
  journal={The Annals of Statistics},
  volume={21},
  number={1},
  pages={217--234},
  year={1993},
  publisher={Institute of Mathematical Statistics}
}

@article{gu2003penalized,
  title={Penalized likelihood density estimation: Direct cross-validation and scalable approximation},
  author={Gu, Chong and Wang, Jingyuan},
  journal={Statistica Sinica},
  pages={811--826},
  year={2003},
  publisher={JSTOR}
}

@article{sun2021asymptotic,
  title={An asymptotic and empirical smoothing parameters selection method for smoothing spline ANOVA models in large samples},
  author={Sun, Xiaoxiao and Zhong, Wenxuan and Ma, Ping},
  journal={Biometrika},
  volume={108},
  number={1},
  pages={149--166},
  year={2021},
  publisher={Oxford University Press}
}

@article{wood2004stable,
  title={Stable and efficient multiple smoothing parameter estimation for generalized additive models},
  author={Wood, Simon N},
  journal={Journal of the American Statistical Association},
  volume={99},
  number={467},
  pages={673--686},
  year={2004},
  publisher={Taylor \& Francis}
}

@book{silverman2018density,
  title={Density Estimation for Statistics and Data Analysis},
  author={Silverman, Bernard W},
  year={2018},
  publisher={Routledge}
}

@article{chen2016comprehensive,
  title={A comprehensive approach to mode clustering},
  author={Chen, Yen-Chi and Genovese, Christopher R and Wasserman, Larry},
  journal={Electronic Journal of Statistics},
  volume={10},
  number={1},
  pages={210--241},
  year={2016},
  publisher={Institute of Mathematical Statistics and Bernoulli Society}
}

@article{perez2009bayesian,
  title={Bayesian classifiers based on kernel density estimation: Flexible classifiers},
  author={P{\'e}rez, Aritz and Larra{\~n}aga, Pedro and Inza, I{\~n}aki},
  journal={International Journal of Approximate Reasoning},
  volume={50},
  number={2},
  pages={341--362},
  year={2009},
  publisher={Elsevier}
}

@article{hinton2012deep,
  title={Deep neural networks for acoustic modeling in speech recognition: The shared views of four research groups},
  author={Hinton, Geoffrey and Deng, Li and Yu, Dong and Dahl, George E and Mohamed, Abdel-rahman and Jaitly, Navdeep and Senior, Andrew and Vanhoucke, Vincent and Nguyen, Patrick and Sainath, Tara N and others},
  journal={IEEE Signal Processing Magazine},
  volume={29},
  number={6},
  pages={82--97},
  year={2012},
  publisher={IEEE}
}

@inproceedings{wolf2020transformers,
  title={Transformers: State-of-the-art natural language processing},
  author={Wolf, Thomas and Debut, Lysandre and Sanh, Victor and Chaumond, Julien and Delangue, Clement and Moi, Anthony and Cistac, Pierric and Rault, Tim and Louf, R{\'e}mi and Funtowicz, Morgan and others},
  booktitle={Proceedings of the 2020 conference on empirical methods in natural language processing: system demonstrations},
  pages={38--45},
  year={2020}
}

@inproceedings{he2016deep,
  title={Deep residual learning for image recognition},
  author={He, Kaiming and Zhang, Xiangyu and Ren, Shaoqing and Sun, Jian},
  booktitle={Proceedings of the IEEE Conference on Computer Vision and Pattern Recognition},
  pages={770--778},
  year={2016}
}

@article{hinton2015distilling,
  title={Distilling the knowledge in a neural network},
  author={Hinton, Geoffrey and Vinyals, Oriol and Dean, Jeff and others},
  journal={arXiv preprint arXiv:1503.02531},
  volume={2},
  number={7},
  year={2015}
}

@inproceedings{yuan2020revisiting,
  title={Revisiting knowledge distillation via label smoothing regularization},
  author={Yuan, Li and Tay, Francis EH and Li, Guilin and Wang, Tao and Feng, Jiashi},
  booktitle={Proceedings of the IEEE/CVF Conference on Computer Vision and Pattern Recognition},
  pages={3903--3911},
  year={2020}
}

@inproceedings{lan2018self,
  title={Self-referenced deep learning},
  author={Lan, Xu and Zhu, Xiatian and Gong, Shaogang},
  booktitle={Asian Conference on Computer Vision},
  pages={284--300},
  year={2018},
  organization={Springer}
}

@inproceedings{phuong2019distillation,
  title={Distillation-based training for multi-exit architectures},
  author={Phuong, Mary and Lampert, Christoph H},
  booktitle={Proceedings of the IEEE/CVF International Conference on Computer Vision},
  pages={1355--1364},
  year={2019}
}

@inproceedings{yang2019snapshot,
  title={Snapshot distillation: Teacher-student optimization in one generation},
  author={Yang, Chenglin and Xie, Lingxi and Su, Chi and Yuille, Alan L},
  booktitle={Proceedings of the IEEE/CVF Conference on Computer Vision and Pattern Recognition},
  pages={2859--2868},
  year={2019}
}

@inproceedings{hou2019learning,
  title={Learning lightweight lane detection {CNNs} by self attention distillation},
  author={Hou, Yuenan and Ma, Zheng and Liu, Chunxiao and Loy, Chen Change},
  booktitle={Proceedings of the IEEE/CVF international conference on computer vision},
  pages={1013--1021},
  year={2019}
}

@article{zhang2020self,
  title={Self-distillation as instance-specific label smoothing},
  author={Zhang, Zhilu and Sabuncu, Mert},
  journal={Advances in Neural Information Processing Systems},
  volume={33},
  pages={2184--2195},
  year={2020}
}

@article{mobahi2020self,
  title={Self-distillation amplifies regularization in {Hilbert} space},
  author={Mobahi, Hossein and Farajtabar, Mehrdad and Bartlett, Peter},
  journal={Advances in Neural Information Processing Systems},
  volume={33},
  pages={3351--3361},
  year={2020}
}

@article{kim2004smoothing,
  title={Smoothing spline Gaussian regression: more scalable computation via efficient approximation},
  author={Kim, Young-Ju and Gu, Chong},
  journal={Journal of the Royal Statistical Society: Series B (Statistical Methodology)},
  volume={66},
  number={2},
  pages={337--356},
  year={2004},
  publisher={Wiley Online Library}
}

@article{huang1998projection,
  title={Projection estimation in multiple regression with application to functional ANOVA models},
  author={Huang, Jianhua Z},
  journal={The Annals of Statistics},
  volume={26},
  number={1},
  pages={242--272},
  year={1998},
  publisher={Institute of Mathematical Statistics}
}

@article{jeon2006effective,
  title={An effective method for high-dimensional log-density ANOVA estimation, with application to nonparametric graphical model building},
  author={Jeon, Yongho and Lin, Yi},
  journal={Statistica Sinica},
  pages={353--374},
  year={2006},
  publisher={JSTOR}
}

@article{lin2006component,
  title={Component selection and smoothing in multivariate nonparametric regression},
  author={Lin, Yi and Zhang, Hao Helen},
  journal={The Annals of Statistics},
  volume={34},
  number={5},
  pages={2272--2297},
  year={2006},
  publisher={Institute of Mathematical Statistics}
}

@book{bosq2012nonparametric,
  title={Nonparametric Statistics for Stochastic Processes: Estimation and Prediction},
  author={Bosq, Denis},
  volume={110},
  year={2012},
  publisher={Springer Science \& Business Media}
}

@book{wahba1990spline,
  title={Spline Models for Observational Data},
  author={Wahba, Grace},
  year={1990},
  publisher={SIAM}
}

@article{sun2016statistical,
  title={Statistical inference for time course RNA-Seq data using a negative binomial mixed-effect model},
  author={Sun, Xiaoxiao and Dalpiaz, David and Wu, Di and S Liu, Jun and Zhong, Wenxuan and Ma, Ping},
  journal={BMC Bioinformatics},
  volume={17},
  number={1},
  pages={1--13},
  year={2016},
  publisher={Springer}
}

@article{helwig2016smoothing,
  title={Smoothing spline analysis of variance models: A new tool for the analysis of cyclic biomechanical data},
  author={Helwig, Nathaniel E and Shorter, K Alex and Ma, Ping and Hsiao-Wecksler, Elizabeth T},
  journal={Journal of Biomechanics},
  volume={49},
  number={14},
  pages={3216--3222},
  year={2016},
  publisher={Elsevier}
}

@article{kimeldorf1971some,
  title={Some results on Tchebycheffian spline functions},
  author={Kimeldorf, George and Wahba, Grace},
  journal={Journal of Mathematical Analysis and Applications},
  volume={33},
  number={1},
  pages={82--95},
  year={1971},
  publisher={Elsevier}
}

@book{wang2011smoothing,
  title={Smoothing Splines: Methods and Applications},
  author={Wang, Yuedong},
  year={2011},
  publisher={CRC press}
}

@book{golub2013matrix,
  title={Matrix Computations},
  author={Golub, Gene H and Van Loan, Charles F},
  year={2013},
  publisher={JHU press}
}

@article{gu1992cross,
  title={Cross-validating non-Gaussian data},
  author={Gu, Chong},
  journal={Journal of Computational and Graphical Statistics},
  volume={1},
  number={2},
  pages={169--179},
  year={1992},
  publisher={Taylor \& Francis}
}

@article{gu1991minimizing,
  title={Minimizing GCV/GML scores with multiple smoothing parameters via the Newton method},
  author={Gu, Chong and Wahba, Grace},
  journal={SIAM Journal on Scientific and Statistical Computing},
  volume={12},
  number={2},
  pages={383--398},
  year={1991},
  publisher={SIAM}
}

@article{hall1990using,
  title={Using the bootstrap to estimate mean squared error and select smoothing parameter in nonparametric problems},
  author={Hall, Peter},
  journal={Journal of Multivariate Analysis},
  volume={32},
  number={2},
  pages={177--203},
  year={1990},
  publisher={Elsevier}
}

@article{wahba1977practical,
  title={Practical approximate solutions to linear operator equations when the data are noisy},
  author={Wahba, Grace},
  journal={SIAM Journal on Numerical Analysis},
  volume={14},
  number={4},
  pages={651--667},
  year={1977},
  publisher={SIAM}
}

@article{xiang1996generalized,
  title={A generalized approximate cross validation for smoothing splines with non-Gaussian data},
  author={Xiang, Dong and Wahba, Grace},
  journal={Statistica Sinica},
  pages={675--692},
  year={1996},
  publisher={JSTOR}
}

@article{wang2018optimal,
  title={Optimal subsampling for large sample logistic regression},
  author={Wang, HaiYing and Zhu, Rong and Ma, Ping},
  journal={Journal of the American Statistical Association},
  volume={113},
  number={522},
  pages={829--844},
  year={2018},
  publisher={Taylor \& Francis}
}

@article{lin2000tensor,
  title={Tensor product space ANOVA models},
  author={Lin, Yi},
  journal={The Annals of Statistics},
  volume={28},
  number={3},
  pages={734--755},
  year={2000},
  publisher={Institute of Mathematical Statistics}
}

@article{daszykowski2002representative,
  title={Representative subset selection},
  author={Daszykowski, Michal and Walczak, Beata and Massart, DL},
  journal={Analytica chimica acta},
  volume={468},
  number={1},
  pages={91--103},
  year={2002},
  publisher={Elsevier}
}

@article{papamakarios2017masked,
  title={Masked autoregressive flow for density estimation},
  author={Papamakarios, George and Pavlakou, Theo and Murray, Iain},
  journal={Advances in Neural Information Processing Systems},
  volume={30},
  year={2017}
}

@article{gao2022adaptive,
  title={Adaptive manifold density estimation},
  author={Gao, Jia-Xing and Jiang, Da-Quan and Qian, Min-Ping},
  journal={Journal of Statistical Computation and Simulation},
  pages={1--15},
  year={2022},
  publisher={Taylor \& Francis}
}

@article{gu2004model,
  title={Model diagnostics for smoothing spline ANOVA models},
  author={Gu, Chong},
  journal={Canadian Journal of Statistics},
  volume={32},
  number={4},
  pages={347--358},
  year={2004},
  publisher={Wiley Online Library}
}

@article{gu2014smoothing,
  title={Smoothing spline ANOVA models: R package gss},
  author={Gu, Chong},
  journal={Journal of Statistical Software},
  volume={58},
  pages={1--25},
  year={2014}
}

@article{gu2001cross,
  title={Cross-validating non-Gaussian data: generalized approximate cross-validation revisited},
  author={Gu, Chong and Xiang, Dong},
  journal={Journal of Computational and Graphical Statistics},
  volume={10},
  number={3},
  pages={581--591},
  year={2001},
  publisher={Taylor \& Francis}
}

@article{strubell2019energy,
  title={Energy and policy considerations for deep learning in NLP},
  author={Strubell, Emma and Ganesh, Ananya and McCallum, Andrew},
  journal={arXiv preprint arXiv:1906.02243},
  year={2019}
}

@book{weinberger1974variational,
  title={Variational Methods for Eigenvalue Approximation},
  author={Weinberger, Hans F},
  year={1974},
  publisher={SIAM}
}

@article{silverman1982est,
  title={On the estimation of a probability density function by the maximum penalized likelihood method},
  author={Silverman, Bernard W},
  journal={The Annals of Statistics},
  pages={795--810},
  year={1982},
  publisher={JSTOR}
}

@article{li1986asymptotic,
  title={Asymptotic optimality of CL and generalized cross-validation in ridge regression with application to spline smoothing},
  author={Li, Ker-Chau},
  journal={The Annals of Statistics},
  pages={1101--1112},
  year={1986},
  publisher={JSTOR}
}

@book{gu2013smoothing,
  title={Smoothing Spline ANOVA Models},
  author={Gu, Chong},
  volume={297},
  year={2013},
  publisher={Springer}
}

@article{craven1978smooth,
  title={Smoothing noisy data with spline functions},
  author={Craven, Peter and Wahba, Grace},
  journal={Numerische Mathematik},
  volume={31},
  number={4},
  pages={377--403},
  year={1978},
  publisher={Springer}
}

@article{nagler2016evading,
  title={Evading the curse of dimensionality in nonparametric density estimation with simplified vine copulas},
  author={Nagler, Thomas and Czado, Claudia},
  journal={Journal of Multivariate Analysis},
  volume={151},
  pages={69--89},
  year={2016},
  publisher={Elsevier}
}

@article{fang2026statistical,
  title   = {A Statistical Perspective on Knowledge Distillation: Foundations, Classical Methods, and LLM Extensions},
  author  = {Fang, Luyang and Lu, Haoran and Cai, Jiazhang and Wang, Tao and Cheng, Huimin and Zhong, Wenxuan and Ma, Ping},
  journal = {Annual Review of Statistics and Its Application},
  year    = {2026},
  note    = {Forthcoming}
}

@article{wahba1985comparison,
  title={A comparison of GCV and GML for choosing the smoothing parameter in the generalized spline smoothing problem},
  author={Wahba, Grace},
  journal={The Annals of Statistics},
  pages={1378--1402},
  year={1985},
  publisher={JSTOR}
}

@book{wand1994kernel,
  title={Kernel smoothing},
  author={Wand, Matt P and Jones, M Chris},
  year={1994},
  publisher={CRC press}
}

@article{bischl2023hyperparameter,
  title={Hyperparameter optimization: Foundations, algorithms, best practices, and open challenges},
  author={Bischl, Bernd and Binder, Martin and Lang, Michel and Pielok, Tobias and Richter, Jakob and Coors, Stefan and Thomas, Janek and Ullmann, Theresa and Becker, Marc and Boulesteix, Anne-Laure and others},
  journal={Wiley Interdisciplinary Reviews: Data Mining and Knowledge Discovery},
  volume={13},
  number={2},
  pages={e1484},
  year={2023},
  publisher={Wiley Online Library}
}

@inproceedings{sandler2018mobilenetv2,
  title={Mobilenetv2: Inverted residuals and linear bottlenecks},
  author={Sandler, Mark and Howard, Andrew and Zhu, Menglong and Zhmoginov, Andrey and Chen, Liang-Chieh},
  booktitle={Proceedings of the IEEE Conference on Computer Vision and Pattern Recognition},
  pages={4510--4520},
  year={2018}
}

@techreport{krizhevsky2009learning,
  title={Learning Multiple Layers of Features from Tiny Images},
  author={Krizhevsky, Alex},
  year={2009},
  institution={University of Toronto}
}

@inproceedings{kolossov2023towards,
  title={Towards a statistical theory of data selection under weak supervision},
  author={Kolossov, Germain and Montanari, Andrea and Tandon, Pulkit},
  booktitle={International Conference on Learning Representations},
  volume={2024},
  pages={41947--41985},
  year={2024}
}

@article{nakkiran2021deep,
  title={Deep double descent: Where bigger models and more data hurt},
  author={Nakkiran, Preetum and Kaplun, Gal and Bansal, Yamini and Yang, Tristan and Barak, Boaz and Sutskever, Ilya},
  journal={Journal of Statistical Mechanics: Theory and Experiment},
  volume={2021},
  number={12},
  pages={124003},
  year={2021},
  publisher={IOP Publishing}
}

@inproceedings{guo2022deepcore,
  title={Deepcore: A comprehensive library for coreset selection in deep learning},
  author={Guo, Chengcheng and Zhao, Bo and Bai, Yanbing},
  booktitle={International Conference on Database and Expert Systems Applications},
  pages={181--195},
  year={2022},
  organization={Springer}
}

@article{yang2022dataset,
  title={Dataset pruning: Reducing training data by examining generalization influence},
  author={Yang, Shuo and Xie, Zeke and Peng, Hanyu and Xu, Min and Sun, Mingming and Li, Ping},
  journal={arXiv preprint arXiv:2205.09329},
  year={2022}
}

@article{sorscher2022beyond,
  title={Beyond neural scaling laws: beating power law scaling via data pruning},
  author={Sorscher, Ben and Geirhos, Robert and Shekhar, Shashank and Ganguli, Surya and Morcos, Ari},
  journal={Advances in Neural Information Processing Systems},
  volume={35},
  pages={19523--19536},
  year={2022}
}

@article{gadre2024datacomp,
  title={Datacomp: In search of the next generation of multimodal datasets},
  author={Gadre, Samir Yitzhak and Ilharco, Gabriel and Fang, Alex and Hayase, Jonathan and Smyrnis, Georgios and Nguyen, Thao and Marten, Ryan and Wortsman, Mitchell and Ghosh, Dhruba and Zhang, Jieyu and others},
  journal={Advances in Neural Information Processing Systems},
  volume={36},
  year={2023}
}

@article{stone1974cross,
  title={Cross-validatory choice and assessment of statistical predictions},
  author={Stone, Mervyn},
  journal={Journal of the Royal Statistical Society: Series B (Methodological)},
  volume={36},
  number={2},
  pages={111--133},
  year={1974},
  publisher={Wiley Online Library}
}

@article{stone1978cross,
  title={Cross-validation: A review},
  author={Stone, Mervyn},
  journal={Statistics: A Journal of Theoretical and Applied Statistics},
  volume={9},
  number={1},
  pages={127--139},
  year={1978},
  publisher={Taylor \& Francis}
}

@incollection{akaike1998information,
  title={Information theory and an extension of the maximum likelihood principle},
  author={Akaike, Hirotogu},
  booktitle={Selected papers of Hirotugu Akaike},
  pages={199--213},
  year={1998},
  publisher={Springer}
}

@article{schwarz1978estimating,
  title={Estimating the dimension of a model},
  author={Schwarz, Gideon},
  journal={The Annals of Statistics},
  pages={461--464},
  year={1978},
  publisher={JSTOR}
}

@article{jones1996survey,
  title={A brief survey of bandwidth selection for density estimation},
  author={Jones, M Chris and Marron, James S and Sheather, Simon J},
  journal={Journal of the American Statistical Association},
  volume={91},
  number={433},
  pages={401--407},
  year={1996},
  publisher={Taylor \& Francis}
}

@article{sheather1991reliable,
  title={A reliable data-based bandwidth selection method for kernel density estimation},
  author={Sheather, Simon J and Jones, Michael C},
  journal={Journal of the Royal Statistical Society: Series B (Methodological)},
  volume={53},
  number={3},
  pages={683--690},
  year={1991},
  publisher={Wiley Online Library}
}

@article{botev2010kernel,
  title        = {Kernel Density Estimation via Diffusion},
  author       = {Botev, Zdravko I. and Grotowski, J. F. and Kroese, Dirk P.},
  journal      = {Annals of Statistics},
  year         = {2010},
  volume       = {38},
  number       = {5},
  pages        = {2916--2957}
}

@article{rudemo1982empirical,
  title        = {Empirical Choice of Histograms and Kernel Density Estimators},
  author       = {Rudemo, Mats},
  journal      = {Scandinavian Journal of Statistics},
  year         = {1982},
  volume       = {9},
  pages        = {65--78}
}

@article{bowman1984alternative,
  title        = {An Alternative Method of Cross-Validation for the Smoothing of Density Estimates},
  author       = {Bowman, Adrian W.},
  journal      = {Biometrika},
  year         = {1984},
  volume       = {71},
  number       = {2},
  pages        = {353--360}
}

@book{fan2018local,
  title={Local polynomial modeling and its applications: monographs on statistics and applied probability 66},
  author={Fan, Jianqing},
  year={2018},
  publisher={Routledge}
}

@book{ferraty2006nonparametric,
  title={Nonparametric functional data analysis: theory and practice},
  author={Ferraty, Fr{\'e}d{\'e}ric and Vieu, Philippe},
  year={2006},
  publisher={Springer}
}

@inproceedings{singh2018machine,
  title={Machine learning-based occupancy estimation using multivariate sensor nodes},
  author={Singh, Adarsh Pal and Jain, Vivek and Chaudhari, Sachin and Kraemer, Frank Alexander and Werner, Stefan and Garg, Vishal},
  booktitle={2018 IEEE Globecom Workshops (GC Wkshps)},
  pages={1--6},
  year={2018},
  organization={IEEE}
}

@misc{lichman2013uci,
author = "Dua, Dheeru and Graff, Casey",
year = "2017",
title = "{UCI} Machine Learning Repository",
url = "http://archive.ics.uci.edu/ml",
institution = "University of California, Irvine, School of Information and Computer Sciences" }

@article{baldi2014searching,
  title={Searching for exotic particles in high-energy physics with deep learning},
  author={Baldi, Pierre and Sadowski, Peter and Whiteson, Daniel},
  journal={Nature communications},
  volume={5},
  number={1},
  pages={1--9},
  year={2014},
  publisher={Nature Publishing Group}
}

@inproceedings{freire2009short,
  title={Short-term memory mechanisms in neural network learning of robot navigation tasks: A case study},
  author={Freire, Ananda L and Barreto, Guilherme A and Veloso, Marcus and Varela, Antonio T},
  booktitle={2009 6th Latin American Robotics Symposium (LARS 2009)},
  pages={1--6},
  year={2009},
  organization={IEEE}
}

@article{ouyang2009chip,
  title={ChIP-Seq of transcription factors predicts absolute and differential gene expression in embryonic stem cells},
  author={Ouyang, Zhengqing and Zhou, Qing and Wong, Wing Hung},
  journal={Proceedings of the National Academy of Sciences},
  volume={106},
  number={51},
  pages={21521--21526},
  year={2009},
  publisher={National Acad Sciences}
}

@article{stoeckius2017large,
  title={Large-scale simultaneous measurement of epitopes and transcriptomes in single cells},
  author={Stoeckius, Marlon and Hafemeister, Christoph and Stephenson, William and Houck-Loomis, Brian and Chattopadhyay, Pratip K and Swerdlow, Harold and Satija, Rahul and Smibert, Peter},
  journal={Nature methods},
  volume={14},
  number={9},
  pages={865},
  year={2017},
  publisher={NIH Public Access}
}

@article{alzahrani2018scraping,
  title={Scraping and preprocessing commercial auction data for fraud classification},
  author={Alzahrani, Ahmad and Sadaoui, Samira},
  journal={arXiv preprint arXiv:1806.00656},
  year={2018}
}

@article{ma2017adaptive,
  title={Adaptive basis selection for exponential family smoothing splines with application in joint modeling of multiple sequencing samples},
  author={Ma, Ping and Zhang, Nan and Huang, Jianhua Z and Zhong, Wenxuan},
  journal={Statistica Sinica},
  pages={1757--1777},
  year={2017},
  publisher={JSTOR}
}

@article{gu2002penalized,
  title={Penalized likelihood regression: General formulation and efficient approximation},
  author={Gu, Chong and Kim, Young-Ju},
  journal={Canadian Journal of Statistics},
  volume={30},
  number={4},
  pages={619--628},
  year={2002},
  publisher={Wiley Online Library}
}

@article{zhang2023optimal,
  title={An optimal transport approach for selecting a representative subsample with application in efficient kernel density estimation},
  author={Zhang, Jingyi and Meng, Cheng and Yu, Jun and Zhang, Mengrui and Zhong, Wenxuan and Ma, Ping},
  journal={Journal of Computational and Graphical Statistics},
  volume={32},
  number={1},
  pages={329--339},
  year={2023},
  publisher={Taylor \& Francis}
}

\end{document}